\documentclass[journal,comsoc]{IEEEtran}
\usepackage{graphicx}
\usepackage{amssymb}
\usepackage{amsmath}
\usepackage{cite}
\usepackage{mathrsfs}
\usepackage{url}
\usepackage[displaymath,mathlines]{lineno}
\usepackage{color}
\usepackage{tabulary}
\usepackage{multirow}
\usepackage{pbox}
\usepackage{multicol}
\usepackage{lipsum}
\usepackage{float}

\usepackage{dcolumn}
\usepackage{booktabs}
\newcolumntype{d}[1]{D..{#1}}
\usepackage{amsthm}
\usepackage{relsize}
\usepackage{lipsum}
\usepackage{epstopdf}
\usepackage{mathtools}% http://ctan.org/pkg/mathtools
\usepackage{calc}% http://ctan.org/pkg/calc
\usepackage{makecell}
\usepackage{caption}
\usepackage{subcaption}
\usepackage{algorithm}
\usepackage{algpseudocode}
\usepackage{tabularx}

\usepackage{booktabs}
\usepackage{siunitx}

\makeatletter
\newcommand{\multiline}[1]{%
  \begin{tabularx}{\dimexpr\linewidth-\ALG@thistlm}[t]{@{}X@{}}
    #1
  \end{tabularx}
}
\makeatother

\algrenewcommand\algorithmicforall{\textbf{foreach}}
\algrenewcommand\algorithmicindent{.8em}
\usepackage{soul}

\makeatletter
\newcommand{\doublewidetilde}[1]{{%
		\mathpalette\double@widetilde{#1}}}
\newcommand{\double@widetilde}[2]{%
		\sbox\z@{$\m@th#1\widetilde{#2}$}%
		\ht\z@=.5\ht\z@
		\widetilde{\box\z@}}
\makeatother

\newtheorem{lemma}{Lemma}

\def\di{\displaystyle}
\newcommand{\tron}[1]{\left(\di #1 \right)} % round bracket
\newcommand{\abs}[1]{\left|#1\right|}
\newcommand{\vuong}[1]{\left[ #1 \right]} % square bracket
\newcommand{\nhon}[1]{\left\{ #1 \right\}} % angle bracket

\begin{document}

\title{\huge Advanced Quantum Annealing for the Bi-Objective Traveling Thief Problem: An $\varepsilon$-Constraint-based Approach}

\author{ Nguyen Hoang Viet, Nguyen Xuan Tung, Trinh Van Chien, \textit{Member, IEEE}, \\ and  Won-Joo Hwang, \textit{Senior Member, IEEE}

\thanks{Nguyen Hoang Viet and Trinh Van Chien are with the School of Information and Communications Technology, Hanoi University of Science and Technology, Hanoi 100000, Vietnam (e-mail: viet.nh220050@sis.hust.edu.vn and chientv@soict.hust.edu.vn). Nguyen Xuan Tung is with the Faculty of Interdisciplinary Digital Technology,
PHENIKAA University, Yen Nghia, Ha Dong, Hanoi 10000, Viet Nam (email: tung.nguyenxuan@phenikaa-uni.edu.vn). Won-Joo Hwang is with the School of Computer Science and Engineering, Center for Artificial Intelligence Research, Pusan National University, Busan 46241, South Korea (e-mail: wjhwang@pusan.ac.kr).  This work was supported in part by BK21 FOUR, Korean Southeast Center for the 4th Industrial Revolution Leader Education; in part 
by the Institute of Information \& communications Technology Planning \& Evaluation (IITP) under the Artificial Intelligence Convergence Innovation Human Resources Development (IITP-2025-RS-2023-00254177) grant funded by the Korea government(MSIT); in part by Quantum Computing based on Quantum Advantage challenge research(RS-2024-00408613) 
through the National Research Foundation of Korea(NRF) funded by the Korean government (Ministry of Science and 
ICT(MSIT)); in part by by Creation of the quantum information science R\&D ecosystem(based on human resources) (Agreement Number) through the National Research Foundation of Korea(NRF) funded by the Korean government (Ministry of Science and ICT(MIST)) (RS-2023-00256050); and in part by the National Research Foundation of Korea(NRF) grant funded by the Korea government(MSIT)(RS-2024-00336962). \textit{Corressponding author:  Won-Joo Hwang}.

}

}



\maketitle

\begin{abstract}
This paper addresses the Bi-Objective Traveling Thief Problem (BI-TTP), a challenging multi-objective optimization problem that requires the simultaneous optimization of travel cost and item profit. Conventional methods for the BI-TTP often face severe scalability issues due to the complex interdependence between routing and packing decisions, as well as the inherent complexity and large problem size. These difficulties render classical computing approaches increasingly inapplicable. To tackle this, we propose an advanced hybrid approach that combines quantum annealing (QA) with the $\varepsilon$-constraint method. Specifically, we reformulate the bi-objective problem into a single-objective formulation by restricting the second objective through adjustable $\varepsilon$-levels, determined within established upper and lower bounds. The resulting subproblem involves a sum of fractional terms, which is reformulated with auxiliary variables into an equivalent form. Subsequently, the equivalent formulation is transformed into a Quadratic Unconstrained Binary Optimization (QUBO) model, enabling direct solution via a quantum annealing (QA) solver. The solutions obtained from the quantum annealer are subsequently refined using a tailored heuristic procedure to further enhance overall performance. By leveraging the flexibility in selecting $\varepsilon$ parameters, our approach effectively captures a broad Pareto front, enhancing solution diversity. Experimental results on benchmark instances demonstrate that the proposed method effectively balances two objectives and outperforms baseline approaches in time efficiency.
\end{abstract}


\vspace{-0.2cm}
\section{Introduction}
% In the field of optimization, most of real-world problems is demonstrated as NP-hard and therefore nontrivial to obtain an optimal solution. Although conventional methods such as branch and bound \cite{xxx} or dynamic programming \cite{xxx} have demonstrated it applicability and power but yet limited to the problem size. Especially, these approaches are operated on classical computers, requiring tremendous time if the instances become bigger due to exponential complexity. Therefore, exact methods become impractical because the sizes of problems in the real world are always enormous. To address it, heuristic approaches are usually proposed \cite{chen2019adaptive,mandziuk2018new,van2024active} for deriving good solutions to speed up the conventional approach in polynomial time. These algorithms usually begin by initializing several valid solutions and then perform several strategies to enhance the quality of the resemblance of solutions throughout iterations. They will terminate after acceptable solutions have been identified, the number of iterations exceeds the maximal loop allowed, or the distinction between solutions of two consecutive iterations is negligible \cite{el2018efficiently}. 

A substantial class of real-world combinatorial optimization problems is NP-hard, making it computationally challenging to find an optimal solution through conventional approaches such as branch-and-bound \cite{malcolm2024multi} and dynamic programming \cite{chauhan2012survey}. To circumvent this limitation, heuristic approaches are commonly employed \cite{kusyk2021survey, stollenwerk2020toward, van2024active}, which, however, require a trade-off between computational efficiency and solution quality. Additionally, real-world problems often exhibit two intertwined characteristics: harsh complexity and interdependent objectives \cite{tung2024jointly}. These features give rise to multi-objective optimization problems (MOPs), where objectives are often in conflict. In such settings, a single solution that optimizes all objectives simultaneously rarely exists. Consequently, solutions are evaluated based on Pareto dominance \cite{wang2018cooperative}, where a solution is considered to dominate another if it is no worse in all objectives and strictly better in at least one. The goal in solving MOPs is to approximate the Pareto front: a set of non-dominated solutions that represent the best trade-offs among objectives.

Many scalarization techniques, such as the weighted-sum method \cite{marler2010weighted} and the $\varepsilon$-constraint method \cite{aghaei2011multi}, have been developed to address multi-objective optimization problems (MOPs). These methods typically convert a MOP into a series of single-objective problems, which are then solved individually to yield high-quality solutions. While the weighted-sum method is limited in its ability to capture non-convex regions of the Pareto front and does not offer control over the solution's position along the front \cite{das1997closer}, the $\varepsilon$-constraint method is highly dependent on the careful selection of the $\varepsilon$-levels \cite{gunantara2018review}. Naive selection can result in infeasible subproblems or poor coverage of the Pareto front. Besides scalarization methods, an alternative and promising approach for MOPs is the use of Multi-Objective Evolutionary Algorithms (MOEAs), due to their ability to simultaneously handle nonlinear, nondifferentiable, and combinatorial objectives within an efficient time \cite{qiao2023self}. Despite their flexibility and robustness, MOEAs often demand a substantial number of function evaluations to converge toward a satisfactory approximation, leading to high computational overhead.

Quantum computing has emerged as a promising paradigm for solving combinatorial optimization problems. Existing approaches can be grouped into three categories based on hardware implementation: (i) purely quantum algorithms that run exclusively on quantum hardware\cite{volpe2025improving,ye2023quantum}; (ii) quantum-inspired heuristics executed on classical machines \cite{gharehchopogh2023quantum, li2020quantum}; and (iii) hybrid methods combining quantum and classical resources \cite{ciacco2025review,yarkoni2022quantum}. Due to the hardware limitations of the noisy intermediate-scale quantum (\texttt{NISQ}) era, hybrid approaches are generally preferred, as they leverage both quantum processing units (QPUs) and state-of-the-art classical solvers for improved scalability in large-scale applications. Among these approaches, hybrid QA stands out due to its practical implementation on the D-Wave quantum annealer \cite{dwaveHybrid} and its seamless integration with classical optimization techniques via the D-Wave Hybrid Solver Service. In this work, we conduct experiments with hybrid QA using the D-Wave Hybrid Solver. 
This solver is capable of solving constrained quadratic optimization problems with up to 500,000 binary variables and 100,000 linear or quadratic constraints. With the support of D-Wave hardware, QA has been demonstrated to achieve superior performance over classical baselines in optimizing single-objective combinatorial problems \cite{marchesin2023improving,guillaume2022deep, wurtz2024solving}. Motivated by these results, we seek to extend this powerful computational paradigm for more challenging scenarios, when the problem are formulated as multi-objective optimization problems. In this work, the BI-TTP is adopted as the case study to evaluate the effectiveness of our proposed approach.

The TTP  \cite{bonyadi2013travelling} is a well-known benchmark that exemplifies the challenges of multi-objective combinatorial optimization. It integrates two classical NP-hard problems, the Traveling Salesman Problem (TSP) and the Knapsack Problem (KP), into a single unified framework. In the TTP, a thief must visit each city exactly once while selecting items of varying value and weight to maximize total profit. However, carrying more items slows down travel speed, creating a conflict between minimizing tour time and maximizing knapsack profit. A faster route limits item collection, while higher profits often require longer travel times. However, several challenges arise when applying QA directly to multi-objective problems such as the BI-TTP. The first limitation is that QA produces only a single solution per execution, necessitating multiple runs to approximate the Pareto front. Consequently, the solution quality may deterioriate if QA repeatedly produces identical solutions. The second challenge lies in the dynamic covariance among objectives, which can significantly degrade the performance of QA if the interference of each objective is not carefully regulated. Under such circumstances, the solution produced by QA may be biased toward a limited region of the objective space, thereby failing to capture the full range of values represented in the true non-dominated solution set. 

Motivated by the aforementioned challenges, this work introduces a novel quantum annealing-based approach for solving the bi-objective Traveling Thief Problem (BI-TTP), specifically under the setting where item value depreciation is omitted. This assumption aligns with the formulations adopted in the BI-TTP competitions held at the \textit{EMO-2019}\footnote{https://www.egr.msu.edu/coinlab/blankjul/emo19-thief/} and \textit{GECCO-2019}\footnote{https://www.egr.msu.edu/coinlab/blankjul/gecco19-thief/} conferences. Our method employs the $\varepsilon$-constraint technique to transform the bi-objective formulation into a single-objective problem, wherein one objective is optimized while the other is bounded by a tunable threshold within predefined upper and lower limits. To exploit the computational advantages of quantum computing, the resulting problem is further reformulated as a QUBO model, enabling direct solution via D-Wave's quantum annealer. A custom-designed heuristic is then applied to refine the raw solutions returned by the quantum hardware. To the best of our knowledge, our method represents the first work to apply quantum computing techniques to solve either the single- or bi-objective variant of the Traveling Thief Problem. The main contributions of this work are summarized as follows:
\begin{itemize}
    \item We propose a novel hybrid approach that combines quantum annealing with the $\varepsilon$-constraint method to the traveling thief problem, where the total item profit is incorporated as a constraint bounded by a tunable threshold. By systematically varying this threshold, we repeatedly minimize the travel cost under different profit constraints, enabling effective exploration of the Pareto front.   
    \item The reformulated single objective problem involves a sum of fractional quadratic terms, which cannot be directly solved using the QA solver. To address this, we propose an iterative transformation technique to handle fractional quadratic terms by reformulating the objective into a difference of quadratic expressions. The resulting problem is then expressed as a QUBO model and efficiently solved using D-Wave's quantum annealer, leveraging the capabilities of quantum annealing to handle large-scale combinatorial optimization. 
    \item To enhance the solutions obtained from D-Wave's quantum annealer, we develop a heuristic refinement procedure that adjusts the item selection. This post-processing step improves solution quality and promotes better convergence toward the Pareto front.
    \item Finally, we evaluate our approach on the synthetic dataset from \cite{polyakovskiy2014comprehensive}, which features three distinct item weight distributions: bounded strongly correlated, uncorrelated, and uncorrelated similar weight. In all categories, our approach consistently outperforms state-of-the-art classical algorithms including MOEA/D \cite{zhang2007moea}, NSGA-II \cite{xu2020multi}, and U-NSGA-III \cite{seada2015unified} in computational efficiency. Specifically, our approach achieves up to an $9\times$ reduction in runtime on the bounded strongly correlated instances.
\end{itemize}

The remainder of this paper is organized as follows. The related works are discussed in Section~\ref{section: related works}  Section~\ref{description} provides a detailed description of the bi-objective Traveling Thief Problem (BI-TTP), along with an illustrative example highlighting its constraints and inherent complexity. Section~\ref{qa-based} presents the overall framework of our proposed approach, which integrates the $\varepsilon$-constraint method with quantum annealing to solve the BI-TTP. In Section~\ref{sol}, we detail the transformation process that reformulates the original problem into a quadratic expression compatible with quantum annealing. Numerical results and performance evaluations are presented in Section~\ref{numerical results}, followed by concluding remarks in Section~\ref{conclusion}.
\vspace{-0.2cm}
\section{Related Works}\label{section: related works}
\textbf{The Traveling Thief Problem.} The TTP was first introduced in \cite{bonyadi2013travelling}, offering a more sophisticated and realistic formulation than the traditional TSP. The TTP is typically formulated in two main variants: a single-objective version (TTP1) and a bi-objective version (TTP2), also referred to as BI-TTP, which serves as the primary focus of this study. The single-objective version (TTP1) combines both objectives into a single fitness function and has been the primary focus of existing researchs \cite{wu2017exact, nikfarjam2024use, el2017local}. In contrast, the bi-objective version (TTP2 or BI-TTP) treats travel time and item value as separate objectives, aiming to determine a set of Pareto-optimal solutions that reveal the trade-off between time and profit. While the original TTP2 formulation in \cite{bonyadi2013travelling} assumes that item values decay over time, this temporal degradation was removed in a later variant proposed in \cite{blank2017solving}. Despite its practical relevance, the BI-TTP has received comparatively less attention in the literature. Some studies have attempted to address this gap. For instance, the authors in \cite{yafrani2017multi} proposed a multi-objective approach to jointly optimize the travel path and item selection, demonstrating solution improvements over single-objective approaches. Similarly, \cite{wu2018evolutionary} employed dynamic programming to identify optimal packing strategies for fixed routes, while \cite{chagas2021non} introduced a specialized multi-objective variant of NSGA-II to approximate the Pareto front. Despite these efforts, existing approaches rely entirely on classical computing paradigms, which struggle with the combinatorial explosion of possibilities inherent in TTP, especially in the bi-objective case. These techniques remain computationally expensive when applied to large-scale or time-sensitive instances, thereby limiting their practical utility.  

\textbf{Quantum Annealing for Combinatorial Optimization problems.} 
Quantum annealing (QA) has increasingly been recognized as a powerful alternative to classical algorithms for solving NP-hard combinatorial optimization problems across various domains.
In finance, \cite{venturelli2019reverse} employed reverse quantum annealing (reverse QA), which integrates both forward and reverse annealing processes to bypass local minima, on the D-Wave 2000Q processor to address the portfolio optimization problem. Nevertheless, hardware constraints restrict its applicability to small-scale instances with at most 64 binary variables. 
In the transportation and logistics domain, QA has also been applied to the air traffic control problem by discretizing delay times into incremental levels and assigning a binary variable to each level to construct a QUBO formulation \cite{stollenwerk2019quantum}. For urban mobility problems, QA has been utilized with manually tuned penalty constraints to prioritize particular assignment decisions \cite{marchesin2023improving}.
Furthermore, \cite{guillaume2022deep} showed that QA significantly outperforms standard mixed-integer linear programming in deep space network scheduling, handling real-world requests much faster. Recent advancements have continued to expand QA's utility: \cite{wurtz2024solving} combined QA with Simulated Annealing to refine solutions for maximum independent set problems, while \cite{kim2025quantum} demonstrated that hybrid QA models can match the accuracy of state of the art classical solvers across more than 50 different NP-hard problems at substantially faster computation time. Additionally, by utilizing a preprocessing protocol to condense network sizes in facility location problems, \cite{ciacco2026quantum} reduced the required binary variables, thereby expanding the scale of models that can be solved.

For multi-objective scenarios, \cite{schworm2024multi} adopted a weighted-sum approach to combine multiple objectives into a single objective function and then obtained solutions through QUBO transformation. However, the weighted-sum method frequently experiences performance drops when dealing with complex intrinsic interference among objectives, primarily because selecting appropriate weights is highly challenging \cite{das1997closer}. In contrast, our approach effectively manages this conflict by imposing strict limits on the values an objective can obtain.  

The classical Traveling Salesman Problem (TSP) has also been efficiently solved using QA, as reported in \cite{le2023quantum}.  More recently, QA has been adapted to the more complex Steiner Traveling Salesman Problem (STSP) by integrating a classical method to filter out unnecessary or costly path distances \cite{ciacco2025steiner}. Nevertheless, its more challenging variant, the Traveling Thief Problem (TTP), has not yet been addressed using QA. This limitation arises because the TTP objective function contains fractional quadratic terms, which cannot be directly solved by QA. In this work, we overcome this limitation by introducing a conversion approach that transforms the TTP objective into an equivalent QUBO formulation suitable for QA.

\vspace{-0.2cm}
\section{Description of the BI-TTP}\label{description}
\begin{figure}[t]
    \centering
    
    \includegraphics[width=1\linewidth]{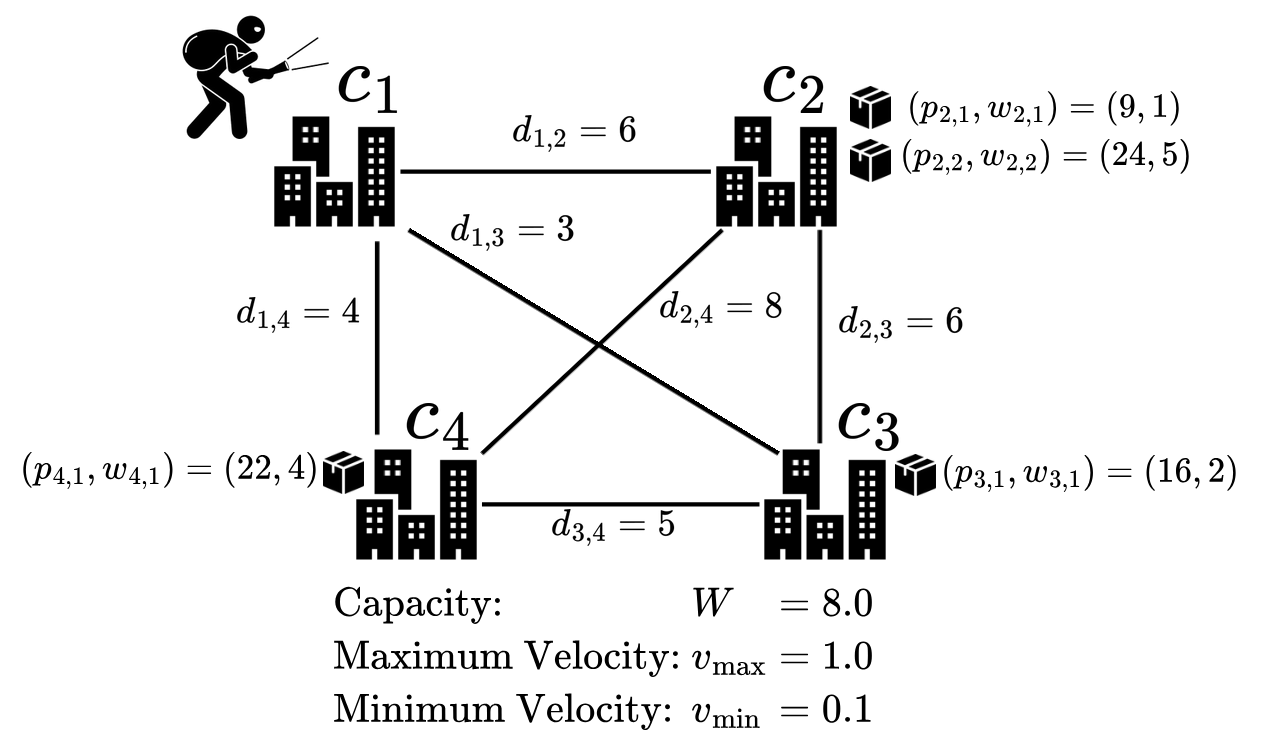}
    \caption{A four-city BI-TTP instance: To complete the tour in the minimum time, the thief can follow the route $c_1 \to c_3 \to c_2 \to c_4 \to c_1$  without collecting any items. Conversely, to maximize profit, the thief must collect items $\mathrm{item}_{2,1}$, $\mathrm{item}_{2,2}$, and $\mathrm{item}_{3,1}$, and travel along the route $c_1 \to c_4 \to c_2 \to c_3 \to c_1$ to complete the tour in the shortest possible time. }
    \label{fig:ttp_example}
    \vspace{-0.2cm}
\end{figure}
In the BI-TTP, a thief must visit a set of $N$ cities $\nhon{c_1, c_2,\ldots,c_N}$, forming a Hamiltonian cycle. The distance between any two cities $c_u$ and $c_v$ is represented by $d_{u,v}$. The knapsack component is integrated into the problem by associating each city $c_i$ with a set of items $B_{i}$. The $k$-th item $\text{item}_{i,k} \in B_{i}$ has an associated profit $p_{i,k}$ and weight $w_{i,k}$. The thief starts and ends the tour at the city $c_1$, and may collect items at any intermediate city $c_2, \ldots, c_N$, subject to a knapsack capacity constraint $W$. That is, the total weight of all selected items must not exceed $W$. Although collecting items increases total profit, it also reduces the thief’s travel speed due to the added weight. When the knapsack is empty, the thief travels at maximum speed $v_{\max}$; when fully loaded, the speed drops to a minimum value $v_{\min}$.

Let the traveling strategy be denoted by $\mathcal{A} = \{\alpha_1, \alpha_2, \ldots, \alpha_N\}$, where $\alpha_1 = 1$ and $(\alpha_2, \ldots, \alpha_N)$ is a permutation of the set $\{2, 3, \ldots, N\}$, representing the remaining $N-1$ cities. In other words, $c_{\alpha_i}$ is the $i$-th city the thief travels in the tour. To facilitate understanding, we also define $c_{\alpha_{N+1}} = c_1$, indicating that the $(N+1)$-th city in the tour is the same as the starting city $c_1$. Besides, we denote picking strategy $\mathcal{Z} = \big\{z_{i,k} \in \{0,1\} \mid k = \overline{1,\abs{B_{i}}} \big\}$, in which $z_{i,k} = 1$ if the thief collects the $k$-th item located in the city $c_i$, and $z_{i,k} = 0,$ otherwise. Along with this notation, the weight of the knapsack after visiting the city $c_{\alpha_i}$ can be calculated as
\begin{equation}\label{curweight}
    W_{i} = \sum\nolimits_{j = 1}^i \sum\nolimits_{k = 1}^{\abs{B_{\alpha_j}}}w_{\alpha_j,k}z_{\alpha_j,k}, \forall \, i = \overline{1,N}.
\end{equation}
where $w_{\alpha_j,k}$ is the weight of the $k$-th item at the city $\alpha_j$. Since items cannot be discarded once collected, the cumulative weight of selected items does not decrease throughout the entire tour. Therefore, we have $W_1 \le W_2 \le \ldots \le W_N$. Additionally, the thief must ensure that the total weight of collected items does not exceed the knapsack capacity $W$. Thus, the capacity constraint of the BI-TTP is represented by 
\begin{equation}
    W_N \le W.
\end{equation}
Consequently, the thief's velocity when traveling from city $c_{\alpha_i}$ to city $c_{\alpha_{i+1}}$ is given by 
\begin{equation}
    v_{i,i+1} = v_{\max} - W_i/W\tron{v_{\max} - v_{\min}},
\end{equation}
This velocity reflects the trade-off between collecting items and minimizing travel time. Then, the time required to travel from city $c_{\alpha_i}$ to city $c_{\alpha_{i+1}}$ can be computed as
\begin{equation}
    t_{i,i+1} = \dfrac{d_{\alpha_i,\alpha_{i+1}}}{v_{i,i+1}} = \dfrac{d_{\alpha_i,\alpha_{i+1}}}{v_{\max} - W_i/W\tron{v_{\max} - v_{\min}}}.
\end{equation}
{Hence, the total duration of the thief's journey is calculated as follows}
\begin{equation}
    f(\mathcal{A}, \mathcal{Z}) = \sum\nolimits_{i=1}^N t_{i,{i+1}}. \label{obj1}
\end{equation}
After visiting all cities, the total profit the thief obtains is calculated as equal to
\begin{equation}
    \sum\nolimits_{i=1}^N \sum\nolimits_{k=1}^{\abs{B_{\alpha_i}}} p_{\alpha_i,k} z_{\alpha_i,k} = \sum\nolimits_{i=1}^N\sum\nolimits_{k=1}^{\abs{B_i}}p_{i,k}z_{i,k},
\end{equation}
where $p_{\alpha_i,k}$ is the profit of the $k$-th item at the city $\alpha_i$. Consequently, the second objective of the BI-TTP is mathematically defined as
\begin{equation}
    g(\mathcal{A}, \mathcal{Z}) = -\sum\nolimits_{i=1}^N \sum\nolimits_{k=1}^{\abs{B_i}} p_{i,k} z_{i,k}, \label{obj2}
\end{equation}
where the negative sign is used to convert the maximization of profit into a minimization objective for consistency with the overall problem formulation. The completed version of BI-TTP is defined as follows
\begin{subequations}\label{prob:mainprob}
    \begin{align}
    \mathcal{A},\mathcal{Z} =& \begin{cases}
        \mathrm{argmin \,\,} f(\mathcal{A}, \mathcal{Z}) \label{prob:mainprob c0}\\
        \mathrm{argmin \,\, } g(\mathcal{A},\mathcal{Z})
    \end{cases} \\
    \text{subject to } & W_N \le W. 
    % \alpha_{i},\alpha_{i+1}\in \mathcal{A}. \label{prob:mainprob c1}
\end{align}
\end{subequations}
The two objectives are inherently conflicting: while collecting more items increases the total reward, it also reduces the thief's velocity, thereby leading to longer travel time. This trade-off necessitates a careful balance between maximizing received profit and minimizing travel time. Fig.~\ref{fig:ttp_example} shows an example, where the thief travels through four cities with different items.
\section{QA-based $\varepsilon$-constraint method}\label{qa-based}
This section proposes a novel approach for solving the BI-TTP by combining QA with the $\varepsilon$-constraint method, termed the QA-based $\varepsilon$-constraint method. We commence by presenting preliminaries on the quantum annealing method, reformulate the problem using binary decision variables, and apply the $\varepsilon$-constraint method to transform this problem into multiple single-objective subproblems.

% A key challenge of the $\varepsilon$-constraint method lies in selecting a feasible scalar $\varepsilon$ for the secondary objective. If $\varepsilon$ falls outside its valid range, the transformed single-objective problem may have no solution. However, estimating the minimum and maximum values of this objective is computationally expensive, as both are NP-hard using traditional methods. 

% The main issue with the $\varepsilon$-constraint method is determining how to select a feasible scalar $\varepsilon$ to constrain the second objective. Indeed, if $\varepsilon$ falls outside the minimum and maximum range of the objective, solutions of the transformed single-objective problem will not exist. However, identifying the approximate minimum and maximum values of the second objective is quite costly, as they both exhibit NP-hard characteristics according to traditional methods. Additionally, each transformed problem is also NP-hard, resulting in a significant amount of time required to fully solve many of the converted problems.
% Meanwhile, the CQM solver can literally provide solutions to these challenges in a short period of time using quantum components. Hence, with the power of quantum characteristic, these above issues can completely circumvent. 
\subsection{Preliminaries on Quantum Annealing}
Quantum annealing (QA) has been widely adopted to solve a broad range of optimization problems. In QA, a quantum processor performs the annealing process on quantum bits (qubits) to minimize an unconstrained Ising model, which is generally formulated as follows:
\begin{equation}\label{ising}
    H_{\text{Ising}}\tron{\sigma} =\sum\nolimits_{i=1}^Mh_i\sigma_i +  \sum\nolimits_{i=1}^M\sum\nolimits_{j=1}^MJ_{ij}\sigma_i\sigma_j,
\end{equation}
where $\sigma_i \in \{-1, 1\}$ represents the spin direction of the $i$-th qubit, and $h_i$ and $J_{ij}$ denote the qubit bias and coupling strength between qubits, respectively. The corresponding quantum Hamiltonian is given by:
\begin{equation}
    H_q\tron{\sigma^x} =\sum\nolimits_{i=1}^Mh_i\sigma_i^x +  \sum\nolimits_{i=1}^M\sum\nolimits_{j=1}^MJ_{ij}\sigma_i^x\sigma_j^x,
\end{equation}
where $\sigma_i^x$ denotes the Pauli-X operator acting on the $i$-th qubit in a Hilbert space of dimension $2^M$~\cite{lucas2014ising}. To obtain the lowest-energy solution, quantum annealing gradually evolves the system's Hamiltonian through a time-dependent annealing schedule, represented as: 
\begin{equation}
    H(t) = \tron{1 - \dfrac{t}{T}}H_0 + \dfrac{t}{T} H_q, \forall \, 0 \le t \le T,
\end{equation}
where $H(t)$ is the Hamiltonian of the quantum system at time $t$, and $T$ is the total duration of the annealing process, $H_0$ and $H_q$ denote the initial and final Hamiltonians (the final Hamiltonian is associated with the considered problem), respectively. According to the quantum adiabatic theorem~\cite{albash2018adiabatic}, if $T$ is sufficiently large and $H_0$ does not commute with $H_q$, the system will remain in its instantaneous ground state throughout the evolution. As a result, the final state of the qubits encodes the global minimum of the original problem. The initial Hamiltonian is typically chosen as $H_0 = - \sum_{i=1}^M \sigma_i^x$ to ensure non-commutativity with $H_q$ and provide a trivial, easily preparable ground state~\cite{rajak2023quantum}.

The QUBO model is a widely adopted framework for representing combinatorial optimization problems as binary quadratic functions. By applying the transformation $x_i = (\sigma_i + 1)/2$, the QUBO formulation corresponding to the Ising model in~\eqref{ising} can be derived as:
\begin{equation}
    H_{\text{QUBO}}\tron{\mathbf{x}} = \sum\nolimits_{i=1}^MQ_{ii}x_i + \sum\nolimits_{i=1}^M\sum\nolimits_{j=1}^MQ_{ij}x_ix_j = \mathbf{x}^T\boldsymbol{Q}\mathbf{x},
\end{equation}
where $\mathbf{x}$ is a binary vector with the $i$-th variable $x_i \in \nhon{0,1}$, and $\boldsymbol{Q} \in \mathbb{R}^{M\times M}$ is the weighted matrix. D-Wave's quantum annealer can solve problems directly when the objective is formulated as a QUBO formulation. However, practical problems naturally involve multiple constraints, which typically take the following form:
\begin{subequations}\label{QA}
    \begin{alignat}{2}
    &\underset{\mathbf{x}}{\mathrm{minimize}} \quad &&H_{\mathrm{cost}}(\mathbf{x}) = \mathbf{x}^TQ\mathbf{x}  \label{QA c0} \\
    &\text{subject to} && H_c^{(j)}(\mathbf{x}) \le 0 , \forall \, j = \overline{1,J}.\label{QA c2}
\end{alignat}
\end{subequations}
The common solution to tackle constraints is to introduce a penalty to the constraints \cite{jeong2025quantum} and solve the obtained problem, 
\begin{equation}
    \underset{\mathbf{x}}{\mathrm{minimize}}\quad H_{\mathrm{QUBO}}\tron{\mathbf{x}} = H_{\mathrm{cost}}(\mathbf{x}) - \sum\nolimits_{j=1}^J \lambda_j H_c^{(j)}(\mathbf{x}). 
\end{equation}
In D-Wave's CQM hybrid solver, penalty coefficients $\lambda_C$ are automatically calibrated, eliminating the need for manual tuning. This allows users to define linear or quadratic objectives and constraints directly, streamlining the formulation of QUBO models and enabling efficient solution retrieval.

\subsection{Binary Reformulation of the BI-TTP}
As QUBO models inherently operate over binary variables, the original problem must be reformulated to express all decision variables in binary form. Therefore, we introduce a binary decision variable $x_{v,i}$, which takes the value $1$ if $c_v$ is the $i$-th city in the tour ($\alpha_i = v$), and $0$ otherwise. Since exactly one city must occupy the $i$-th position in the tour, the following two constraints must be satisfied
\begin{align}
    \sum\nolimits_{v=1}^N x_{v,i} = 1, \forall\, i = \overline{1,N};\quad \sum\nolimits_{i=1}^N x_{v,i} = 1, \forall\, v = \overline{1,N}\label{sumcon2},
\end{align}
which dictate that for all $1 \le v,i,j \le N$ where $i \ne j$, the condition $x_{v,i} = x_{v,j} = 1$ is strictly prohibited. By ensuring that no city is visited twice, our binary formulation inherently eliminates the possibility of subtours.
Accordingly, the distance from city $\alpha_i$ to $\alpha_{i+1}$ is calculated  as follows
\begin{equation}
    d_{\alpha_i, \alpha_{i+1}} =\sum\nolimits_{u=1}^N\sum\nolimits_{v=1}^N d_{u,v}x_{u,i}x_{v,i+1}. 
\end{equation}
For the given constraints in \eqref{sumcon2}, there exists a unique tuple $(u, v, i)$ such that $x_{u,i} = x_{v,i+1} = 1$, indicating that city $c_u$ is assigned to position $i$ and city $c_v$ to position $i+1$ in the tour. 
Furthermore, using the binary variables $x_{v,i}$, the knapsack weight after visiting city $c_{\alpha_i}$, as defined in~\eqref{curweight}, can be reformulated as:
\begin{equation}
    W_i = \sum\nolimits_{j=1}^i \sum\nolimits_{t=1}^N \tron{x_{t,j}\sum\nolimits_{k=1}^{\abs{B_t}}z_{t,k}w_{t,k} },
\end{equation}
where $\di\sum\nolimits_{k=1}^{\abs{B_t}}z_{t,k}w_{t,k}$ represents the total weight of all items picked at city $c_t$. To simplify the representation, we define the column vectors $\mathbf{x}_v = [x_{v,1},\ldots,x_{v,i},\cdots, x_{v,N}]^T$ and $\mathbf{z}_i = [z_{i,1},\ldots, z_{i,k},\ldots z_{i,|B_i|}]^T$, which are then aggregated into the matrices $\mathbf{X} = [\mathbf{x}_{1},\ldots,\mathbf{x}_{v},\cdots, \mathbf{x}_{N}]$ and $\mathbf{Z} = [\mathbf{z}_{1},\ldots,\mathbf{z}_{i},\cdots, \mathbf{z}_{N}]$. The first objective function $f\tron{\mathbf{X}, \mathbf{Z}}$ is 
\begin{align}\label{obj:travel}
   f\tron{\mathbf{X}, \mathbf{Z}} &= \sum\nolimits_{i=1}^Nt_{i,i+1}
   = \sum\nolimits_{i=1}^N\dfrac{\displaystyle\sum\nolimits_{u =1}^N\sum\nolimits_{v=1}^Nd_{u,v}x_{u,i}x_{v,i+1}}{v_{i,i+1}}\notag\\
&=\sum\nolimits_{i=1}^N\dfrac{\displaystyle\sum\nolimits_{u =1}^N\sum\nolimits_{v=1}^Nd_{u,v}x_{u,i}x_{v,i+1}}{v_{\max} - W_i/W(v_{\max} - v_{\min})} \notag \\
     &= \sum\nolimits_{i=1}^N\dfrac{W\displaystyle\sum\nolimits_{u=1}^N\sum\nolimits_{v=1}^Nd_{u,v}x_{u,i}x_{v,i+1}}{Wv_{\max} - W_i(v_{\max}-v_{\min})}
\end{align}
% denote the binary vectors formed by the variables $x_{v,i}, \forall \, 1 \le v, i \le n$ and $z_{i,k}, \forall \, 1 \le i \le n, 1 \le k \le |B_i|$, respectively.
Similarly, the second objective $g\tron{\mathbf{X}, \mathbf{Z}}$ is computed as
\begin{equation}
    g\tron{\mathbf{X}, \mathbf{Z}} = -\sum\nolimits_{i=1}^N \sum\nolimits_{t=1}^N \tron{x_{t,i}\sum\nolimits_{k=1}^{\abs{B_t}}z_{t,k}p_{t,k}},
\end{equation}
where $\displaystyle\sum\nolimits_{k=1}^{|B_t|} z_{t,k} p_{t,k}$ represents the total profit accumulated by the thief after visiting city $c_t$. The BI-TTP formulation in~\eqref{prob:mainprob} can now be expressed as: 
\begin{subequations}\label{prob:reformulate}
    \begin{alignat}{2}
    &\mathrm{minimize} \quad &&f\tron{\mathbf{X}, \mathbf{Z} }=\sum\nolimits_{i=1}^N \dfrac{W\displaystyle\sum\nolimits_{u=1}^N\sum\nolimits_{v=1}^Nd_{u,v}x_{u,i}x_{v,i+1}}{Wv_{\max} - W_i(v_{\max}-v_{\min})}   \label{prob:reformulate c0}\\
    & \, &&g\tron{\mathbf{X}, \mathbf{Z}} = -\sum\nolimits_{i=1}^N \sum\nolimits_{t=1}^N \tron{x_{t,i}\sum\nolimits_{k=1}^{\abs{B_t}}z_{t,k}p_{t,k}} \label{prob:reformulate c1}\\
    &\text{subject to} && \sum\nolimits_{i=1}^N \sum\nolimits_{t=1}^N \tron{x_{t,i}\sum\nolimits_{k=1}^{\abs{B_t}}z_{t,k}w_{t,k} }\le W \label{prob:reformulate c2} \\
& \, && x_{t,i} \in \nhon{0,1}, \forall \, t,i = \overline{1,N} \label{prob:reformulate c3} \\
&\, &&\sum\nolimits_{v=1}^N x_{v,i} = 1,\forall \, i = \overline{1,N} \label{prob:reformulate c4}\\
&\, &&     \sum\nolimits_{j=1}^N x_{v,i} = 1, \forall \, v = \overline{1,N} \label{prob:reformulate c5}\\ 
& \, && z_{t,k} \in \nhon{0,1}, \forall \, t = \overline{1,N}, k = \overline{1,\abs{B_t}} \label{prob:reformulate c6}
\end{alignat}
\end{subequations}
The reformulated problem in~\eqref{prob:reformulate} is now expressed as a binary optimization problem; however, it still involves the simultaneous optimization of two conflicting objectives, making it challenging to obtain optimal solutions. To address this, the next subsection introduces the use of the $\varepsilon$-constraint method, which transforms the problem into a series of more tractable single-objective subproblems.

% Constraints $\eqref{prob:reformulate c2} - \eqref{prob:reformulate c6}$ are all requirements of the BI-TTP, in which constraint $\eqref{prob:reformulate c2}$ ensures the total weight of all collected items does not exceed the capacity of the knapsack, constraints $\eqref{prob:reformulate c3} - \eqref{prob:reformulate c5}$ ensure the tour of the thief will form a Hamiltonian cycle, and constraint $\eqref{prob:reformulate c6}$ indicates that the thief can choose to either pick an item or not. 
\vspace{-0.2cm}
\subsection{Dispensing with $\varepsilon$-constraint method}
\begin{algorithm}[t]
    \caption{QA-based $\varepsilon$-constraint method for the BI-TTP}
    \label{alg:flow}
\hspace*{\algorithmicindent}\textbf{Input:}
The binary version of the BI-TTP, i.e. problem \eqref{prob:reformulate}; the number of equal segments $S$.  \\
    \hspace* {\algorithmicindent}\textbf{Output:} The Pareto front of the BI-TTP
    \begin{algorithmic}[1]
        \State Identify the minimum and the maximum of the second objective \eqref{prob:reformulate c1} without concerning about the first objective \eqref{prob:reformulate c0} using QA, denoted as $g_{\mathrm{min}}$ and $g_{\mathrm{max}}$, respectively. 
        \For {$0 \le s < S$}
            \State $\varepsilon_s \gets g_{\min} + s\cdot (g_{\max} - g_{\min})/{S}$
        \EndFor
        \For {$0 \le s < S$}
            \State \multiline{Generate the $s$-th $\varepsilon$-constraint subproblem by constraining the objective \eqref{prob:reformulate c1} such that $\varepsilon_s \le g(\mathbf{X}, \mathbf{Z}) \le \varepsilon_{s+1}$. } 
            \State \multiline{Obtain the solution $(\mathbf{X}_s, \mathbf{Z}_s)$ of the $s$-th subproblem using QA.} 
        \EndFor
        \State Determine the non-dominated set from identified solutions $(\mathbf{X}_s, \mathbf{Z}_s), \forall \, 0 \le s < S$ as the Pareto front. 
    \end{algorithmic}
\end{algorithm}
To solve the binary formulation of the BI-TTP, we propose a novel approach called the QA-based $\varepsilon$-constraint method, which combines QA with the $\varepsilon$-constraint technique. Specifically, we apply the $\varepsilon$-constraint method to transform the bi-objective problem in~\eqref{prob:reformulate} into a sequence of single-objective subproblems, each of which is solved using quantum annealing. To enable this transformation, the first step involves determining the range of possible knapsack profit values. Let $g_{\min}$ and $g_{\max}$ denote the minimum and maximum achievable profits, respectively, computed without considering the travel time objective. The detailed methodology for identifying these values is given in Subsection~\ref{Subsection: Max-min 2nd identification}. 

Given the determined range of knapsack profit, the interval $\left[g_{\min}, g_{\max}\right]$ is divided into $S$ segments $[\varepsilon_s, \varepsilon_{s+1}]$, where $g_{\min} = \varepsilon_0 \le \varepsilon_1 \le \ldots \le \varepsilon_S = g_{\max}$. In our experiments, we initialize $\varepsilon_s = g_{\min} + s \cdot (g_{\max} - g_{\min})/S,\forall \, 0 \le s \le S$ such that each interval $\left[\varepsilon_s, \varepsilon_{s+1}\right]$ covers an equal portion of the value spectrum. For the given knapsack constraint value corresponding to the $s$-th segment, the associated subproblem is formulated as follows:
\begin{subequations}\label{prob:subprob}
    \begin{alignat}{2}
    &\mathrm{minimize} \quad&&  f\tron{\mathbf{X}, \mathbf{Z} }=\sum\nolimits_{i=1}^N \dfrac{W\displaystyle\sum\nolimits_{u=1}^N\sum\nolimits_{v=1}^Nd_{u,v}x_{u,i}x_{v,i+1}}{Wv_{\max} - W_i(v_{\max}-v_{\min})} \label{prob:subprob c0}\\
    &\text{subject to} && \varepsilon_s \le g\tron{\mathbf{X},\mathbf{Z}} \le \varepsilon_{s+1}, \label{prob:subprob c1}\\
   &\, && \eqref{prob:reformulate c2}-\eqref{prob:reformulate c6}, \nonumber
\end{alignat}
\end{subequations}
The selection of the number of segments is critical, as if the interval width is too narrow, the feasible region of the corresponding subproblem may become empty, making it unsolvable. Conversely, if the interval is too large, the constraint may become ineffective, reducing its impact on guiding the search. Additionally, each interval must lie within the range of achievable knapsack profits, i.e., $\left[g_{\min}, g_{\max}\right]$, to ensure feasibility. Consequently, in our approach, we divide this interval uniformly into $S$ equal segments, which facilitates uniform coverage of the Pareto front across the objective space. This systematic segmentation also enables effective boundary exploration, leveraging the flexibility of tuning $\varepsilon$, in contrast to the stochastic nature of MOEAs that rely on randomly initialized populations.

The goal in solving the problem in \eqref{prob:subprob} is to minimize the total travel time while satisfying a constraint on the knapsack profit. However, the resulting objective function is expressed as a sum of fractional terms, which poses challenges for direct optimization. To address this, we propose an iterative algorithm that reformulates the objective into an equivalent form suitable for quantum annealing. The detailed transformation and solution procedure are described in Subsection~\ref{subproblem: s-th segment subproblem}. After solving all subproblems, the final Pareto front is constructed by extracting the set of non-dominated solutions from the union of all obtained results. The overall framework of the proposed approach is summarized in Algorithm~\ref{alg:flow}.

\vspace{-0.2cm}
\section{Finding solution to the BI-TTP}\label{sol}
This section presents the solution procedures for: $(i)$ identifying the upper and lower bounds of the knapsack profit; $(ii)$ minimizing the total travel time under constrained knapsack profit values; and $(iii)$ enhancing the solutions returned by the quantum annealer through a tailored heuristic refinement. %Finally, we illustrate the overall complexity of the QA-based $\varepsilon$-constraint method in the context of classical simulation and its variable usage. 

% We implement our method to address the BI-TTP. Initially, we determine the range of possible profit values that the thief could achieve. This range is subsequently divided into equal segments, with the two endpoints of each segment introducing a constraint in the new subproblem. Since the subproblem’s objective is not quadratic, we propose a fractional programming approach to establish its equivalence with a reformulated problem and solve the latter iteratively. Finally, we derive the time complexity of our approach in the context of classical simulation. 

\subsection{Identification of Minimum and Maximum Knapsack Values}\label{Subsection: Max-min 2nd identification}
% We now need to calculate the minimum value of the second objective, which is equivalent to the maximum profit that the thief can receive. In other words, we will identify the minimum value of $g\tron{\mathbf{x}, \mathbf{z}}$ without any interference with the first objective $f\tron{\mathbf{x}, \mathbf{z}}$. 
The problem of identifying the lower bound of $g\tron{\mathbf{X}, \mathbf{Z}}$ is
\begin{subequations}\label{prob:knapsack}
    \begin{alignat}{2}
    &\mathrm{minimize} \quad && g(\mathbf{X},\mathbf{Z}) = -\sum\nolimits_{i=1}^N \sum\nolimits_{t=1}^N \tron{x_{t,i}\sum\nolimits_{k=1}^{\abs{B_t}}z_{t,k}p_{t,k}} \label{prob:knapsack c0}\\
    &\text{subject to} && \eqref{prob:reformulate c2} - \eqref{prob:reformulate c6}. \nonumber
\end{alignat}
\end{subequations}
It is evident that both the objective function and constraints in problem~\eqref{prob:knapsack} are formulated as linear or quadratic expressions over binary variables, making the problem directly compatible with the D-Wave CQM solver for efficient solution search. In contrast, the maximum value of $g(\mathbf{X}, \mathbf{Z})$ is $0$ when the thief does not collect any items along the tour, i.e., $z_{t,k} = 0$ for all $t \in \{1,\ldots, N\}$ and $k \in \{1,\ldots, |B_t|\}$.

% The optimal solution produced by the quantum annealer is then treated as a lower bound for the second objective  $g(\mathbf{X}, \mathbf{Z})$.

\subsection{The Travel Time Minimization with Constrained Knapsack's Values}\label{subproblem: s-th segment subproblem}
The travel time minimization problem, subject to the knapsack value constrained within the segment $[\varepsilon_s, \varepsilon_{s+1}]$, is rewritten as follows:
\begin{equation}\label{prob: s-th segment subproblem - full}
    \begin{aligned}
    &\mathrm{minimize} &&  f\tron{\mathbf{X}, \mathbf{Z} }=\sum\nolimits_{i=1}^N \dfrac{W\displaystyle\sum\nolimits_{u=1}^N\sum\nolimits_{v=1}^Nd_{u,v}x_{u,i}x_{v,i+1}}{Wv_{\max} - W_i(v_{\max}-v_{\min})} \\
    &\text{subject to} && \varepsilon_s \le  g\tron{\mathbf{X},\mathbf{Z}} \le \varepsilon_{s+1}, \\
    & \, && \eqref{prob:reformulate c2} - \eqref{prob:reformulate c6}.
    % &\, && \sum_{i=1}^N \sum_{t=1}^N \tron{x_{t,i}\sum_{k=1}^{\abs{B_t}}z_{t,k}w_{t,k} }\le W  \\
    % & \, && x_{t,i} \in \nhon{0,1}, \forall \, t,i = \overline{1,N}  \\
    % &\, &&\sum_{v=1}^N x_{v,i} = 1,\forall \, i = \overline{1,N}  \\
    % &\, && \sum_{j=1}^N x_{v,i} = 1, \forall \, v = \overline{1,N} \\ 
    % & \, && z_{t,k} \in \nhon{0,1}, \forall \, t = \overline{1,N}, k = \overline{1,\abs{B_t}} 
\end{aligned}
\end{equation}
To enable compatibility with quantum annealing, both the objective and the constraints must be reformulated into a QUBO structure. However, the objective function in \eqref{prob: s-th segment subproblem - full} contains fractional terms, which fall outside the standard quadratic form. To address this challenge, we propose an iterative transformation technique. Specifically, we construct an equivalent formulation of the subproblem \eqref{prob: s-th segment subproblem - full}, grounded in the results established by Lemma~\ref{Lemma:convert}, as presented below:
\begin{lemma}\label{Lemma:convert}
Given the subproblem in \eqref{prob: s-th segment subproblem - full}, we represent it in a more concise form as follows:
\begin{equation}\label{prob:brief_prob}
    \begin{aligned}
    &\underset{\mathbf{X}, \mathbf{Z}}{\mathrm{minimize}} \quad && \sum\nolimits_{i=1}^N\dfrac{A_i(\mathbf{X},\mathbf{Z})}{B_i(\mathbf{X},\mathbf{Z})} \\
    &\text{subject to} && (\mathbf{X},\mathbf{Z}) \in \mathcal{D} 
\end{aligned}
\end{equation}
where $A_i\left(\mathbf{X}, \mathbf{Z}\right) = W\displaystyle\sum\nolimits_{u=1}^N\sum\nolimits_{v=1}^N d_{u,v} x_{u,i} x_{v,i+1}$ and $B_i \left(\mathbf{X}, \mathbf{Z}\right) = Wv_{\max} - W_i (v_{\max} - v_{\min}), \forall \, i = \overline{1,N}$. Here, $\mathcal{D}$ denotes the feasible domain defined by all constraints in problem~\eqref{prob: s-th segment subproblem - full}. This problem is equivalent to 
\begin{equation}\label{prob:convert}
   \begin{aligned}
    &\underset{\mathbf{X}, \mathbf{Z}, \mathbf{b}}{\mathrm{minimize}} \quad && \sum\nolimits_{i=1}^N\dfrac{A_i(\mathbf{X}, \mathbf{Z})}{b_i}\\
    &\text{subject to} && B_i(\mathbf{X}, \mathbf{Z}) \ge b_i > 0, \forall \, i = \overline{1,N},\\
    & \, && (\mathbf{X}, \mathbf{Z}) \in \mathcal{D},
\end{aligned} 
\end{equation}
where $\mathbf{b} = \vuong{b_1,b_2,\ldots,b_n}^T$. 
\end{lemma} 
\begin{proof}
    See Appendix~\ref{appendxi: proof convert}.  
\end{proof}
By applying Lemma~\ref{Lemma:convert}, the objective is now correctly formulated within the constrained QUBO framework. Therefore, we can now focus on finding optimal solutions of the converted problem \eqref{prob:conv}, which is represented as follows:
\begin{equation}\label{prob:conv}
    \begin{aligned}
    &\mathrm{minimize} &&  f\tron{\mathbf{X}, \mathbf{Z}, \mathbf{b}} = \sum\nolimits_{i=1}^N \dfrac{W\displaystyle\sum\nolimits_{u=1}^N\sum\nolimits_{v=1}^Nd_{u,v}x_{u,i}x_{v,i+1}}{b_i}\\
    &\text{subject to} && Wv_{\max} - W_i(v_{\max}-v_{\min}) \ge b_i, \forall \, i = \overline{1,N}, \\
    &&& b_i > 0, \forall \, i = \overline{1,N}, \\
    &&&  \varepsilon_s \le g\tron{\mathbf{X},\mathbf{Z}} \le \varepsilon_{s+1}, \\
    & \, && \eqref{prob:reformulate c2} - \eqref{prob:reformulate c6}.
\end{aligned}
\end{equation}
To comprehensively solve the problem in \eqref{prob:conv}, we adopt an iterative approach to alternately optimize the decision variables $\left(\mathbf{X}, \mathbf{Z}\right)$ and the auxiliary variable $\mathbf{b}$.
\subsubsection{Update solution $(\mathbf{X}, \mathbf{Z})$ with fixed scalar $\mathbf{b}$}
With the scalar variable $\mathbf{b}$ held constant, this step involves solving for the variables $(\mathbf{X}, \mathbf{Z})$ using the CQM solver.
Specifically, at the $t$-th iteration, the problem in~\eqref{prob:conv} becomes
\begin{equation}\label{prob:conv fix b}
    \begin{aligned}
    &\mathrm{minimize} &&  f^{(t)}\tron{\mathbf{X}, \mathbf{Z}, \mathbf{b}} = \sum_{i=1}^N \dfrac{W\displaystyle\sum\nolimits_{u=1}^N\sum\nolimits_{v=1}^Nd_{u,v}x_{u,i}x_{v,i+1}}{b^{(t-1)}_i} \\
    &\text{subject to} && Wv_{\max} - W_i(v_{\max}-v_{\min}) \ge b^{(t-1)}_i, \forall \, i = \overline{1,N}, \\
    &&& \varepsilon_s \le g\tron{\mathbf{X},\mathbf{Z}} \le \varepsilon_{s+1}, \\
    & \, && \eqref{prob:reformulate c2} - \eqref{prob:reformulate c6},
\end{aligned}
\end{equation}
where $\mathbf{b}^{(t-1)}$ is taken from the previous iteration. 
The problem conforms to the QUBO model structure and can be efficiently solved using the CQM solver, yielding a low-energy solution $\left(\mathbf{{X}}^{\text{QA},(t)}, \mathbf{{Z}}^{\text{QA},(t)}\right)$ at the $t$-th iteration. 

Note that current quantum processors operate within the noisy intermediate-scale quantum (\textsf{NISQ}) era, where limited coherence times and hardware noise can degrade solution quality \cite{ohkura2022simultaneous}. As a result, quantum annealing may yield suboptimal solutions due to noise-induced errors or by becoming trapped in local optima while exploring the feasible space. To mitigate these limitations and improve solution quality, we propose a heuristic post-processing operator, termed \textit{Later and Enough Accumulation} (LEA).

Observing that selecting items from cities later in the tour keeps the knapsack lighter during earlier moments of the tour, thereby increasing the thief's velocity and reducing overall travel time. Therefore, we propose a selection strategy biased toward accumulating items near the end of the tour. Moreover, the LEA heuristic is designed to guarantee feasibility with respect to the total profit constraint. In particular, the solution $(\mathbf{X}^{\text{LEA}}, \mathbf{Z}^{\text{LEA}})$ is refined so that $g(\mathbf{X}^{\text{LEA}}, \mathbf{Z}^{\text{LEA}})$ is closest to the corresponding threshold $\varepsilon_s$ in the $s$-th subproblem. As a result, this design naturally aligns with the $\varepsilon$-constraint method, while it is not readily compatible with other classical approaches. A detailed description of this operator can be found in Appendix~\ref{appendix: LEA description}.

\subsubsection{Update auxiliary variable $\mathbf{b}$ based on received $(\mathbf{X}, \mathbf{Z})$}
\begin{algorithm}[t]
    \caption{Iterative Approach for problem \eqref{prob:conv}}
    \label{alg:iterative}
\hspace*{\algorithmicindent}\textbf{Input:}
The distance from two arbitrary cities $u$ and $v$: $d_{u,v}$ with $u,v = \overline{1,N}$; the bag's capacity $W$; the maximum and minimum velocity of the thief: $v_{\text{max}}$ and $v_{\text{min}}$; the profit of $k$-th item at city $t$: $p_{t,k}$ with $t = \overline{1,N}$, $k = \overline{1,\abs{B_t}}$; two endpoints of segment $\varepsilon_s, \varepsilon_{s+1}$ and auxiliary variable $\mathbf{b}$; 
the maximum number of iterations $T$. \\
    \hspace*{\algorithmicindent}\textbf{Output:} The traveling strategy $\mathbf{X}$ and the picking plan $\mathbf{Z}$
    \begin{algorithmic}[1]
        \State Set the iteration index $t \gets 0$, the auxiliary vector $\mathbf{b}^{(t)} \gets \mathbf{1}_N^T$, where $\mathbf{1}_N$ is a vector having $N$ elements equal to $1$. 
        \While{$t < T$} 
            \State $t \gets t + 1$
            \State Use QA to identify solution $\tron{\mathbf{X}^{\mathrm{QA},(t)}, \mathbf{Z}^{\mathrm{QA},(t)}}$ of problem \eqref{prob:conv}. 
            \State $\tron{\mathbf{X}^{(t)}, \mathbf{Z}^{(t)}}\gets \mathrm{LEA}\tron{\mathbf{X}^{\mathrm{QA},(t)}, \mathbf{Z}^{\mathrm{QA},(t)}}$
            \State Obtain $\mathbf{b}^{(t)}$ by \eqref{b_i-update}
            \If{$f(\mathbf{X}^{(t)}, \mathbf{Z}^{(t)}, \mathbf{b}^{(t)}) < f(\mathbf{X}^{(t-1)}, \mathbf{Z}^{(t-1)}, \mathbf{b}^{(t-1)})$}
                \State \textbf{break}
            \EndIf
            
        \EndWhile
        \State \Return $\mathbf{X}_s = \mathbf{X}^{(t-1)}, \mathbf{Z}_s = \mathbf{Z}^{(t-1)}$
    \end{algorithmic}
\end{algorithm}
Given the enhanced solution $\left(\mathbf{{X}}^{(t)}, \mathbf{{Z}}^{(t)}\right)$, the auxiliary variable $\mathbf{b}^{(t)}$ is updated as follows:
\begin{equation}\label{b_i-update}
    b_i = Wv_{\max} - W_i(v_{\max}-v_{\min}), \forall \, i = \overline{1,N},
\end{equation}
where this update guarantees the largest reduction of the objective function. The iterative process of updating $\mathbf{b}$ and solving for $\left(\mathbf{X}, \mathbf{Z}\right)$ continues until convergence as in Lemma~\ref{coroll1}.
 
\begin{lemma}\label{coroll1}
Given any initial point $(\mathbf{X}^{(0)}, \mathbf{Z}^{(0)}, \mathbf{b}^{(0)})$ that satisfies the feasibility conditions of problem~\eqref{prob:conv}, the iterative procedure, solving the subproblem in~\eqref{prob:conv fix b} to update $(\mathbf{X}, \mathbf{Z})$ and then updating $\mathbf{b}$ using~\eqref{b_i-update}, progressively narrows the search space and guarantees convergence.
\end{lemma}
\begin{proof}
    See Appendix~\ref{appendix: proof coroll1}. 
\end{proof}

Based on the above formulation, the complete framework for solving problem~\eqref{prob:conv} is outlined in Algorithm~\ref{alg:iterative}. 
Initially, the auxiliary variable $\mathbf{b}$ is set to $\mathbf{1}^T$, which provides a small and positive starting point preserving the feasibility of problem \eqref{prob:conv}. During each iteration, the decision variables $\mathbf{X}$ and $\mathbf{Z}$, along with $\mathbf{b}$, are updated in an alternating fashion. Specifically, a candidate solution $(\mathbf{X}, \mathbf{Z})$ is first obtained using the CQM solver. This solution is then refined via the proposed LEA heuristic to improve quality. Subsequently, with $(\mathbf{X}, \mathbf{Z})$ fixed, the auxiliary variable $\mathbf{b}$ is updated using the rule defined in equation~\eqref{b_i-update}. The process repeats until either the objective value fails to improve or a predefined iteration limit is reached.
% \textcolor{blue}{ Turning to the number of constraints, we consider those defined in \eqref{prob:conv fix b}, which are directly solved by QA in our framework. Specifically, constraints \eqref{prob:reformulate c2}, \eqref{prob:reformulate c4}, and \eqref{prob:reformulate c5} contribute $1,N$, and $N$ constraints, respectively, while there are three additional constraints in \eqref{prob:conv fix b}. Therefore, the total number of constraints in each QA iteration is $2N + 4$.
% }
\vspace{-0.2cm}
\subsection{Variable Utilization and Computational Complexity}
Since D-Wave hardware restricts the number of variables to at most 500000, it is crucial to carefully analyze variables consumption when encoding the BI-TTP into a QUBO formulation. For the binary tour variable $x_{v,i}$ with $1 \le v,i \le N$, we have utilized $N\times N = N^2$ variables. Similarly, for the picking variable $z_{i,k}$ with $1 \le i \le N, 1 \le k \le \abs{B_i}$, the encoding requires $N \times \di\max_{1 \le i \le N}\abs{B_i}$ variables. Therefore, the total number of variables used in our encoding is $N \times (N + \max_{1 \le i \le N}\abs{B_i})$.

Regarding the time complexity, for a problem involving $K$ binary variables, the theoretical time complexity of quantum annealing (QA) is $\mathcal{O}(e^{\sqrt{K}})$, as established in \cite{mukherjee2015multivariable}. In our formulation, the traveling strategy $\mathbf{X}$ and the item selection plan $\mathbf{Z}$ require $\mathcal{O}(N^2)$ and $\mathcal{O}(NM)$ binary variables, respectively, where $N$ is the number of cities and $M$ is the total number of items. Therefore, the time complexity for computing the minimum and maximum bounds of the second objective is $\mathcal{O}(e^{\sqrt{N^2 + NM}})$. Given that the range of the second objective is divided into $S$ equal segments, $S$ subproblems must be solved. For each subproblem, each iteration involves one QA call and one application of the LEA heuristic, with complexities of $\mathcal{O}(e^{\sqrt{N^2 + NM}})$ and $\mathcal{O}(M^2)$, respectively.
After obtaining $S$ solutions from the subproblems, identifying the non-dominated set to construct the Pareto front incurs a computational cost of $\mathcal{O}\left(S \log S\right)$, as discussed in \cite{mishra2010fast}. Therefore, the overall computational complexity of the proposed algorithm is $\mathcal{O}\left(STe^{\sqrt{n^2 + nm}} + STm^2 + S \log S\right)$, where $T$ denotes the number of iterations required for convergence. 

It is important to note that the aforementioned time complexity is derived under the assumption of classical simulation of quantum annealing. In contrast, our proposed approach utilizes D-Wave's Leap Service, which employs a quantum-classical hybrid CQM solver that integrates advanced classical heuristics with quantum processing capabilities. By exploiting quantum phenomena such as superposition and entanglement, the QPU can simultaneously explore a vast portion of the solution space, thereby enabling near-optimal solutions to be found in significantly reduced time. A comprehensive performance evaluation is provided in Section~\ref{numerical results}.  
\vspace{-0.2cm}
\section{Numerical Results}\label{numerical results}
\begin{table}[t]
\centering
\caption{Information of considered instances}
\label{table:instances}
\resizebox{0.8\linewidth}{!}{
\begin{tabular}{ccccc}
\midrule
Instances & $n$ & $m$              & $W$                         & Knapsack Type \\  \midrule
\texttt{ch150\_n149\_bsc}   & 150        & \multicolumn{1}{l}{149} & 12310                    & bsc           \\
\texttt{ch150\_n149\_unc}   & 150        & \multicolumn{1}{l}{149} & 6860                    & unc           \\
\texttt{ch150\_n149\_usw}   & 150        & \multicolumn{1}{l}{149} & 13608                    & usw           \\ \midrule
\texttt{a280\_n279\_bsc}   & 280        & \multicolumn{1}{l}{279} & 25936                    & bsc           \\
\texttt{a280\_n279\_unc}   & 280        & \multicolumn{1}{l}{279} & 12718                     & unc           \\
\texttt{a280\_n279\_usw}   & 280        & \multicolumn{1}{l}{279} & 25478                    & usw           \\ \midrule
\end{tabular}
}
\end{table}
This section presents experimental results validating the effectiveness of the proposed approach using diverse benchmark instances to ensure generality across various settings.% Moreover, the numerical results demonstrate that our proposed hybrid framework consistently outperform state-of-the-art classical multi-objective optimization methods.
\subsection{Experimental Setup}
\textbf{Environment.} All experiments were conducted on Kaggle using an Intel(R) Xeon(R) CPU (4 cores, 2.20 GHz) with 30 GB RAM. The implementation environment includes Python 3.10, Dimod 0.12, D-Wave System 1.32, and Pymoo 0.6. All experiments were conducted on the D-Wave Advantage2 Quantum Processing Units (QPUs), which support up to 4,400 qubits and 40,000 couplers. The Advantage2 QPU employs the Zephyr topology, in which each qubit is connected to 16 orthogonal neighboring qubits via internal couplers, along with 2 external couplers and 2 odd couplers that link it to similarly oriented qubits. Furthermore, the QUBO formulation is implemented using D-Wave’s Constrained Quadratic Model (CQM) framework and solved through the Leap Service’s hybrid quantum–classical CQM solver. Consequently, the reported runtime corresponds to the total time CQM solver spent working on each problem, encompassing QPU access time and classical postprocessing time.

\textbf{Datasets.} We employ the instances \texttt{ch150\_n149} and \texttt{a280\_n279} from \cite{polyakovskiy2014comprehensive}, which differ in the number of cities and items. Notably, the \texttt{a280\_n279\_bcs} instance has been widely adopted in BI-TTP competitions at \textit{EMO-2019}\footnote{\url{https://www.egr.msu.edu/coinlab/blankjul/emo19-thief/}}, \textit{GECCO-2019}\footnote{\url{https://www.egr.msu.edu/coinlab/blankjul/gecco19-thief/}}, and \textit{GECCO-2024}\footnote{\url{https://sites.google.com/view/ttp-gecco2024/}}. Detailed characteristics of the test cases are provided in Table~\ref{table:instances}, where $N$, $M$, and $W$ denote the number of cities, items, and knapsack capacity, respectively. Item weights follow three standard distributions: bounded strongly correlated (\texttt{bsc}), uncorrelated (\texttt{unc}), and uncorrelated with similar weights (\texttt{usw}). For main results, the number of $\varepsilon$-segments $S$ is set to $10$. 
 
\begin{figure*}
\centering
\includegraphics[width=\linewidth]{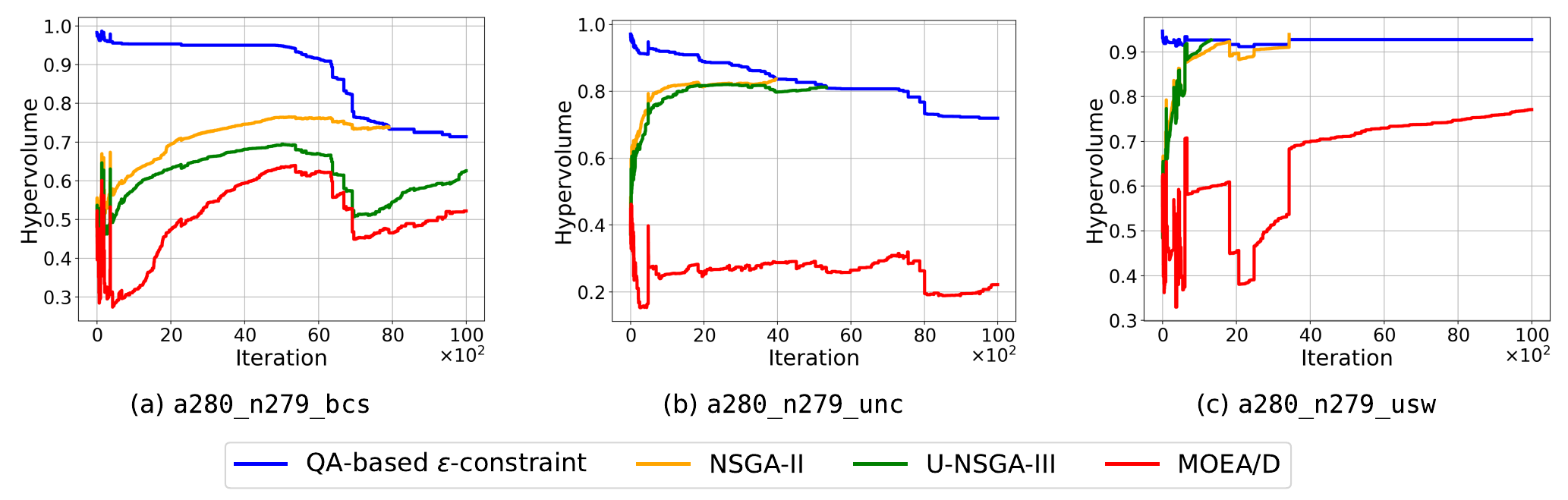}
\caption{\label{fig:conv} The HV of QA-based $\varepsilon$-constraint method, NSGA-II, U-NSGA-III, and MOEA/D at each iteration corresponding to the instance \texttt{a280\_n279}. Our method produces a complete and high-quality Pareto front via Quantum Annealing, after which the evolving solution sets of NSGA-II, U-NSGA-III, and MOEA/D are incrementally compared at each iteration. }
\end{figure*}

\textbf{Baselines.} For comparison, we evaluate the performance of our approach again state-of-the-art classical MOEAs widely used for solving bi-objective optimization problems: \textbf{NSGA-II} \cite{xu2020multi} and \textbf{U-NSGA-III} \cite{seada2015unified}. Both NSGA-II and U-NSGA-III extend the genetic algorithm (GA) framework \cite{bu2022cognitive} by employing non-dominated sorting to address multi-objective optimization, focusing on approximating the Pareto front rather than optimizing a single objective. While retaining the core GA operations, crossover and mutation, they differ in selection mechanisms: NSGA-II employs non-dominated sorting and crowding distance for forming the next population, whereas U-NSGA-III utilizes a set of reference directions to encourage uniform Pareto front distribution. We further consider the \textbf{MOEA/D} algorithm \cite{zhang2007moea}, an alternative multi-objective evolutionary algorithm that employs a distinct decomposition strategy to transform a multi-objective problem into multiple single-objective subproblems. We implement three algorithms using the Pymoo framework \cite{blank2020pymoo}, which streamlines the development of evolutionary optimization methods. Following the configuration in \cite{chagas2021non}, we set the population size to $500$, the crossover rate to $0.8$, and the mutation rate to $0.1$. We run these baselines for  10000 iterations on the \texttt{a280\_n279} instance and 4000 iterations on \texttt{ch150\_n149}.

\begin{figure*}
\centering
\includegraphics[width=\linewidth]{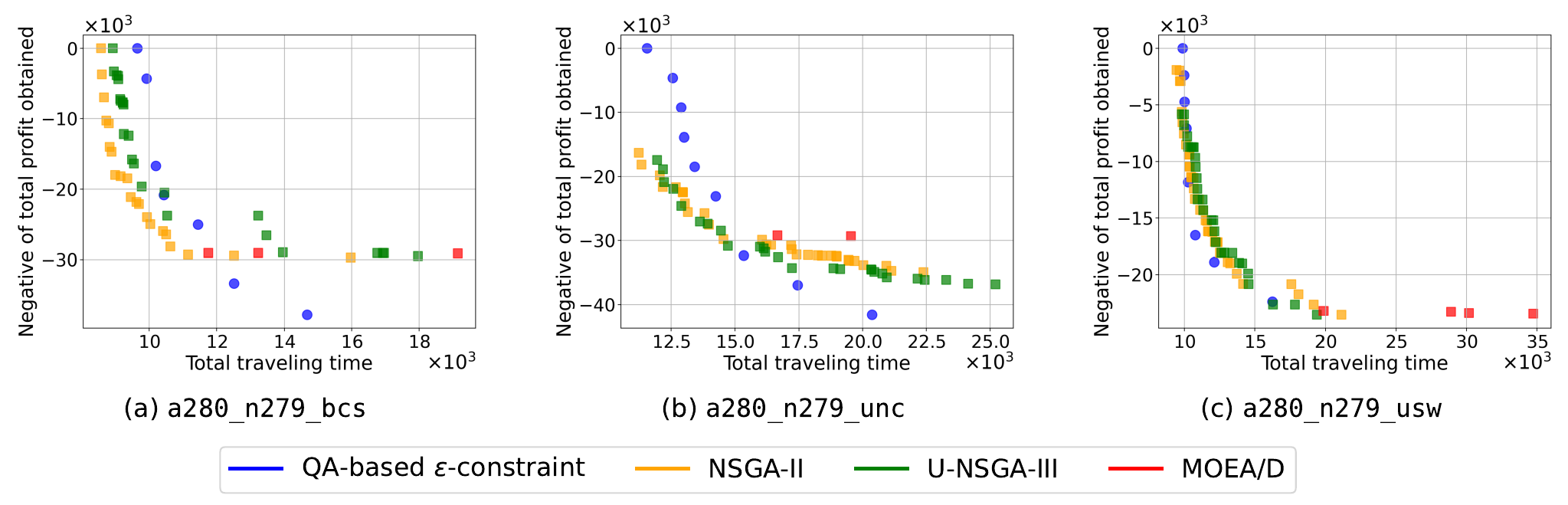}
\caption{\label{fig:res}
Comparison of Pareto fronts produced by NSGA-II, U-NSGA-III, MOEA/D, and the QA-based $\varepsilon$-constraint method for the instance \texttt{a280\_n279}. Our approach consistently yields well-distributed solutions along the second objective across all instances, while NSGA-II and U-NSGA-III, show considerable variation across instances as a result of their reliance on the initialization process. Furthermore, MOEA/D fails to adequately explore the feasible region and frequently collapses to nearly identical values on the second objective.
}
\end{figure*}

\textbf{Metrics.} We utilize Hypervolume (HV) and running time to evaluate our method for different methods. HV is a widely adopted metric for evaluating the quality of Pareto fronts in multi-objective optimization, as it reflects how close the solutions are to the true Pareto front and the diversity of the non-dominated solutions obtained \cite{riquelme2015performance}. A crucial element in HV calculation is the reference point. Following the protocols established in \textit{GECCO-2019} and \textit{EMO-2019}, the reference point is chosen based on all outcomes from compared algorithms. Particularly,
we normalize all considered points via min-max scaling and use the reference point $[1,1]^T$, which represents the worst scaled outcome for each objective. The HV is hence calculated as the overall area dominated by all points in the Pareto front with respect to the chosen reference point. In terms of computational time, that of the classical baseline algorithms are recorded end-to-end, while the runtime of the proposed hybrid approach is obtained by aggregating the execution times returned by the D-Wave framework.

\vspace{-0.2cm}
\subsection{Comparison with State-of-the-art methods}
Fig.~\ref{fig:conv} illustrates the evolution of the hypervolume (HV) for different methods. Since our quantum-based method generates a complete, high-quality Pareto front in a single run, we incrementally merge the evolving solutions from NSGA-II and U-NSGA-III with our fixed Pareto front at each iteration to enable a fair comparison. As a result, Fig.~\ref{fig:conv} exhibits that the HV curve of our method presents a slight decreasing trend, not due to solution degradation, but due to the integration of intermediate solutions from the evolving fronts of the evolutionary algorithms. Across all three subfigures, the hypervolume (HV) of our quantum-based method consistently exceeds $0.7$, indicating a well-distributed Pareto front that effectively dominates a substantial portion of the feasible region, which is widely regarded as a hallmark of high-quality solutions in multi-objective optimization \cite{riquelme2015performance}. Notably, in subfigure~(\ref{fig:conv}a), U-NSGA-III requires after 10000 iterations cannot match the HV achieved by our method in a single quantum-assisted run. Meanwhile, although U-NSGA-III gradually converges toward comparable HV levels, it remains less stable in reaching competitive front quality, as illustrated in subfigures~(\ref{fig:conv}b) and (\ref{fig:conv}c). By contrast, MOEA/D struggles to achieve comparable solution quality on all benchmark instances. This limitation arises because the MOEA/D framework decomposes a multi-objective problem into multiple subproblems using scalar objective weights, which makes it difficult to capture the complex interactions between objectives in the BI-TTP.

\begin{table*}[]
\centering
\resizebox{\textwidth}{!}{
\begin{tabular}{c|cccccc}
\toprule
& \texttt{ch150\_n149\_bcs} & \texttt{ch150\_n149\_unc} & \texttt{ch150\_n149\_usw} & \texttt{a280\_n279\_bcs} & \texttt{a280\_n279\_unc} & \texttt{a280\_n279\_usw} \\ \midrule Total QPU time
& $1.2938 \pm 0.0690$         & $1.2993 \pm 0.1280$         & $1.3336 \pm 0.0054$         & $1.4070 \pm 0.0395$        & $1.3991 \pm 0.0325$        & $1.3362 \pm 0.0010$        \\ \midrule
MOEA/D                                                                              & $\mathbf{9678 \pm 197}$      & $\mathbf{10079 \pm 657}$      & $\mathbf{9718\pm 273}$      & $\mathbf{38170 \pm 884}$    & $\mathbf{41614 \pm 265}$    & $\mathbf{40672 \pm 967}$    \\
NSGA-II                                                                             & $5431 \pm 239$ ($7.5\times$)      & $1708 \pm 498$ ($2.1\times$)       & $917 \pm 57$ ($1.3\times$)      & $33771 \pm 4018$ ($8.1\times$)   & $16907 \pm 3191$ ($4.0\times$)     & $13229 \pm 2553$ ($3.3\times$)    \\
U-NSGA-III                                                                          & $\mathbf{8712 \pm 135}$      & $1365 \pm 85$ ($1.9\times$)       & $966 \pm 155$ ($1.3\times$)     & $\mathbf{38977 \pm 732}$    & $17900 \pm 4583$ ($4.2\times$)    & $8484 \pm 2714$ ($2.1\times$)   \\ \midrule
Ours (without LEA) & $724 \pm 8$       & $719 \pm 5$       & $726 \pm 5$       & $4144 \pm 42$     & $4203 \pm 10$     & $3998 \pm 142$     \\ 
Ours (with LEA)    & $727 \pm 9$       & $723 \pm 5$       & $728 \pm 5$       & $4161 \pm 41$     & $4227 \pm 11$     & $4006 \pm 142$     \\ \bottomrule
\end{tabular}
}
\caption{Comparison of computational time between the proposed method and three baseline algorithms, namely MOEA/D, NSGA-II, and U-NSGA-III. The row labeled “Total QPU time” reports the quantum processing unit (QPU) usage time of our approach, whereas all remaining entries correspond to the overall runtime. For the baseline methods, runtime is measured either at the iteration where their solutions are compared with those produced by our approach or upon algorithm termination. The rows “Ours (without LEA)” and “Ours (with LEA)” denote the proposed method evaluated without and with the LEA heuristic, respectively. \textbf{Bolded values} indicate that the corresponding algorithm reached its termination criterion. Values shown in parentheses $(a\times)$ represent the computational overhead relative to our method. All reported results are averaged over three independent runs to ensure fair comparison. }
\label{table:qa run time}
\vspace{-0.2cm}
\end{table*}

Fig.~\ref{fig:res} illustrates the Pareto fronts generated by different methods for the \texttt{a280\_n279} instance. The privileged advantage of the proposed QA-based $\varepsilon$-constraint method is its systematic strategy for covering the full range of the second objective by dividing it into uniformly spaced intervals. This structured exploration generates well-distributed solutions along the axis of the second objective, enabling the resulting Pareto front to span a wide spectrum of trade-offs. In contrast, NSGA-II and U-NSGA-III rely heavily on the initial population and their evolutionary mechanisms, limiting the capability of consistently exploring the entire objective space. As shown in Fig.~\ref{fig:res}, their performance varies significantly across different instances. For example, all algorithms perform relatively well on the \texttt{a280\_n279\_usw} instance, likely due to the uniformity in item weights, which simplifies the trade-off between travel time and item selection. However, on the more complex \texttt{a280\_n279\_bsc} instance, where item weights are strongly correlated with profits and require more sophisticated decision-making, both NSGA-II and U-NSGA-III fail to cover the lower end of the second objective spectrum, leaving solutions with values below $-30,000$ undominated. Similarly, for the \texttt{a280\_n279\_unc} instance, the Pareto fronts obtained by NSGA-II and U-NSGA-III are concentrated around the central region of the feasible space, failing to capture extreme trade-offs near the lower and upper bounds of the second objective.

Table~\ref{table:qa run time} reports the runtime statistics (in seconds) for all compared algorithms across various benchmark instances. For our proposed QA-based $\varepsilon$-constraint method, we separately list the total quantum processing unit (QPU) time and the overall execution time, which includes data preprocessing, classical optimization steps prior to QPU invocation, and the application of the LEA heuristic for post-solution refinement. The reported results show that our approach consistently finds high-quality Pareto fronts faster than classical methods across all BI-TTP instances. Notably, the total runtime and QPU access time for each instance are almost equal, as both are primarily governed by the CQM solver, depending on problem size, input variables, and constraint complexity.  In contrast, the classical methods exhibit high runtime variability, as their performance is highly sensitive to the choice of initialization.

Interestingly, in the strongly correlated bounded item distribution setting, U-NSGA-III is likewise unable to match the solution quality achieved by our approach. More importantly, our method also achieves solutions up to $8.1$ times faster than NSGA-II on the \texttt{a280\_n279\_bcs} benchmark, which is recognized as one of the standard benchmarks for algorithm comparison in the BI-TTP competitions at \textit{EMO-2019}\footnote{https://www.egr.msu.edu/coinlab/blankjul/emo19-thief/}, \textit{GECCO-2019}\footnote{https://www.egr.msu.edu/coinlab/blankjul/gecco19-thief/}, and \textit{GECCO-2024}\footnote{\url{https://sites.google.com/view/ttp-gecco2024/}}.  Furthermore, the runtime advantage over classical baselines becomes increasingly pronounced for the \texttt{a280\_n279} instance compared to \texttt{ch150\_n149}, as the computational overhead relative to our proposed method consistently increases across all instances, thereby underscoring the superior scalability of the proposed approach on larger problem instances. Consequently, QA is able to produce high-quality solutions that state-of-the-art classical baselines may fail to obtain, highlighting the effectiveness of quantum hardware in exploiting the overall feasible search space within a limited time budget. Finally, the runtime overhead associated with LEA at all benchmarks constitutes only up to $0.5\%$ of the total execution time of the proposed approach and is therefore negligible.

% \textcolor{blue}{The hypervolume (HV) of all obtained solution sets is reported in Table~\ref{table:qa hv}. For each benchmark instance, the reference point is selected from the union of non-dominated solution sets produced by all compared algorithms. Under this reference-point selection, the HV values for the \texttt{bcs} and \texttt{unc} instances are substantially higher, whereas a slightly lower HV is observed for the \texttt{usw} instance. This behavior can be attributed to the fact that, for the \texttt{bcs} and \texttt{unc} instances, our method is able to cover a broad range of values along the second objective, while exhibiting a modest degradation in coverage for the remaining instance, as visualized in Fig.~\ref{fig:res}. Since our approach is designed to consistently spans a wide spectrum of objective values, its HV remains stable across different benchmarks. Furthermore, despite introducing only negligible computational overhead, the LEA effectively repairs solutions obtained from the D-Wave sampler, leading to a significant improvement in HV performance.
% } 

\vspace{-0.2cm}
\subsection{Ablation Study}
This section focuses on analyzing the efficacy of the proposed QA-based $\varepsilon$-constraint method. We first demonstrate the advantages of the $\varepsilon$-constraint utilization in our hybrid framework over the weighted-sum method. Next, we examine the contribution of the LEA heuristic to solution improvement.  Finally, we
analyze the quality of the obtained solutions under different values of the parameter $S$ and varying $\varepsilon$-levels. All experiments are performed on the instance \texttt{a280\_n279}.

\begin{table}[]
\center
\resizebox{\columnwidth}{!}{
\begin{tabular}{c|ccc}
\toprule
 & \texttt{a280\_n279\_bcs} & \texttt{a280\_n279\_unc} & \texttt{a280\_n279\_usw} \\ \midrule
\makecell{Weighted-sum \\(w/o LEA)}                                         & $0.64 \pm 0.02$    & $0.54 \pm 0.22$    & $0.55 \pm 0.01$    \\ \hline
\makecell{$\varepsilon$-constraint \\(w/o LEA)}      & $0.71 \pm 0.02$    & $0.73 \pm 0.03$     & $0.75\pm 0.06$    \\ \hline
 \makecell{$\varepsilon$-constraint \\(with LEA)}    & $0.89 \pm 0.03$    & $0.89\pm 0.04$     & $0.87 \pm 0.02$    \\ \bottomrule
\end{tabular}
}
\caption{
Comparison of hypervolume (HV) obtained by different approaches for addressing the bi-objective BI-TTP. 
The “Weighted-sum (w/o LEA)” row represents the scalarized formulation of the two objectives in problem~\eqref{prob:reformulate}. 
The “$\varepsilon$-constraint” rows correspond to our proposed framework, evaluated with and without the LEA heuristic.
}
\label{table:hv ablation}
\end{table}

To demonstrate the advantage of explicitly controlling the second objective via the $\varepsilon$-constraint framework, we compare our method against a weighted-sum approach. In particular, the objectives in problem~\eqref{prob:reformulate} are aggregated into a single scalar objective given by $\alpha\cdot f(\mathbf{X}, \mathbf{Z}) + (1-\alpha)\cdot g(\mathbf{X}, \mathbf{Z})$. To ensure a fair comparison with $S$ subproblems, the weight parameter is chosen as $\alpha_k = k/(S-1)$ for $k = 0, \ldots, S-1$. Moreover, as the LEA heuristic is tailored specifically to the $\varepsilon$-constraint method and is not applicable to the weighted-sum formulation, it is excluded from this experiment. As shown in Table~\ref{table:hv ablation}, the proposed $\varepsilon$-constraint framework achieves up to $35\%$ higher solution quality than the weighted-sum approach under the same number of runs. This performance gap arises because the weighted-sum formulation lacks explicit control over the second objective. Notably, the variance observed for the instance \texttt{a280\_n279\_unc} is relatively large, as the item weights are independently distributed, further highlighting the importance of explicitly controlling one objective when addressing complicated multi-objective problems.

\begin{figure}
    \centering
    \includegraphics[width=1\linewidth]{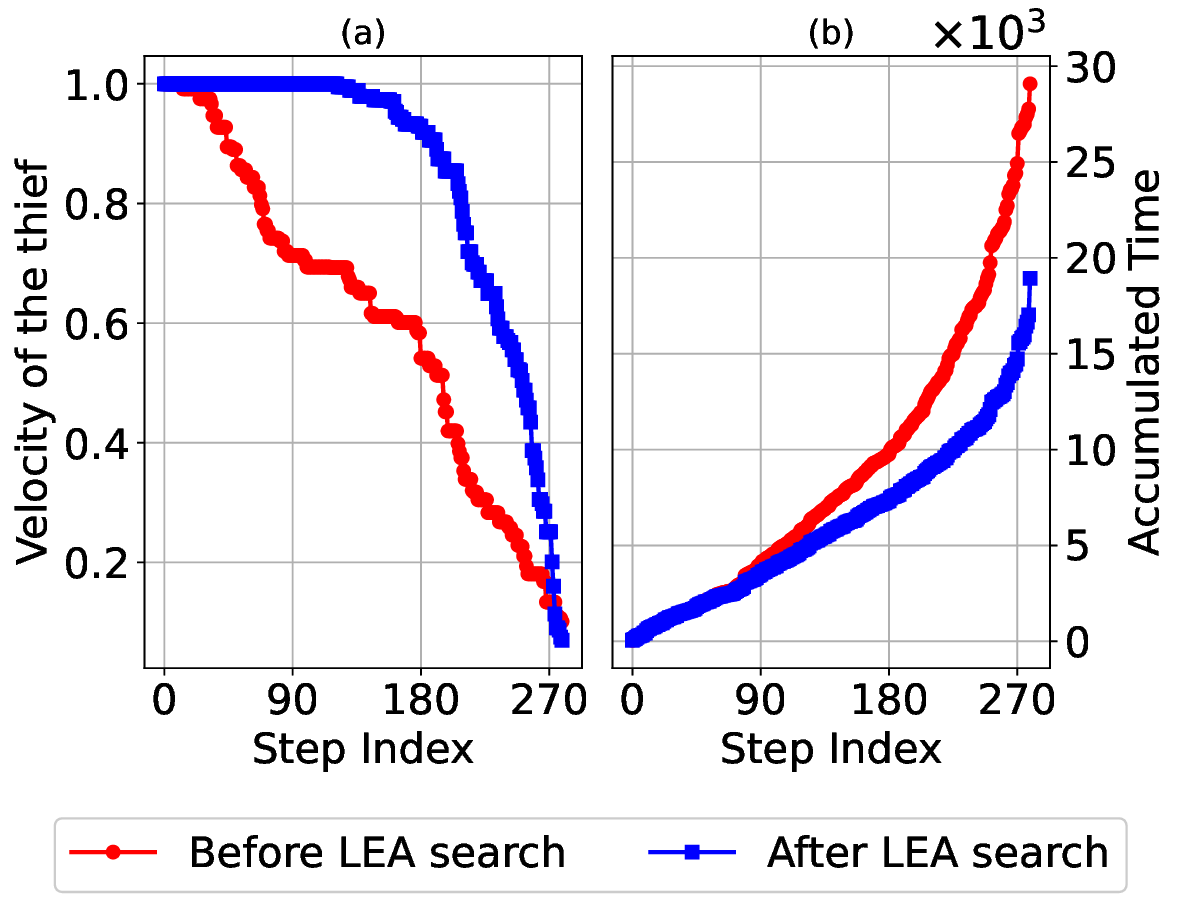}
    \caption{The change in solution before and after applying LEA search}
    \label{fig:LEAcompare}
\end{figure}
We now further examine the impact of the proposed LEA heuristic on the overall performance of our approach. As reported in Table~\ref{table:hv ablation}, despite its negligible computational overhead, LEA contributes significantly to improving solution quality at all benchmarks. To provide a clearer illustration, Fig.~\ref{fig:LEAcompare} presents a comparison of the thief’s velocity and total travel time for a representative Pareto-optimal solution from the \texttt{a280\_n279\_bcs} instance, specifically the solution with the lowest travel time. As shown in subfigure~(\ref{fig:LEAcompare}a), the thief initially moves at maximum velocity, but without LEA (red line), the velocity rapidly decreases due to early item collection. In contrast, with LEA (blue line), velocity remains higher and more stable after the first third of the tour. This results in significantly reduced travel time, as shown in subfigure (\ref{fig:LEAcompare}b), where the total duration drops from approximately 29,000 seconds to 18,000 seconds, an improvement of nearly 40\%. 

\begin{figure}
    \centering
    \includegraphics[width=\linewidth]{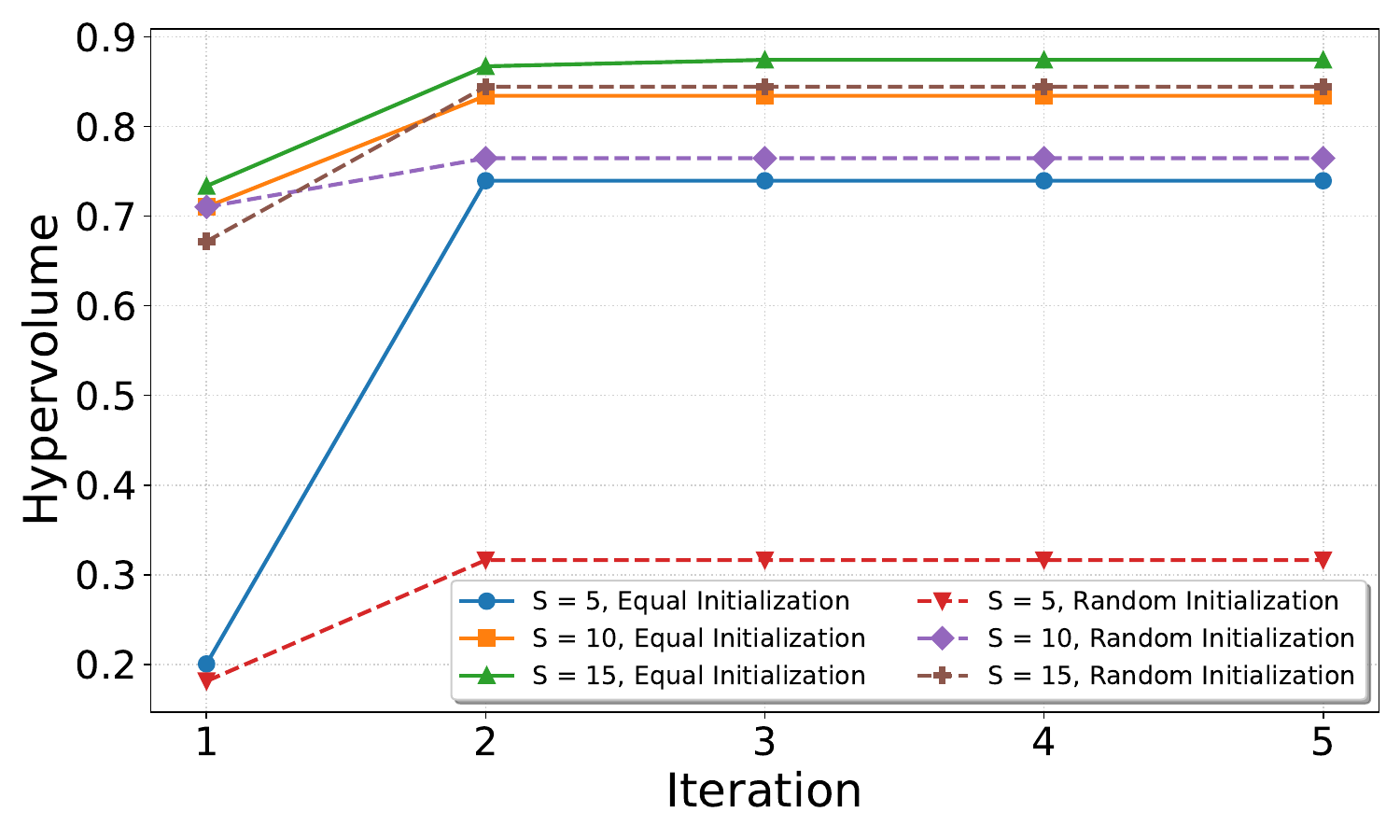}
    \caption{Sensitivity analysis of the QA-based $\varepsilon$-constraint method with respect to the number of segments $S$ and the initialization of the $\varepsilon$ levels. }
\label{table: epsilon level}
\end{figure}

We conduct a detailed ablation study to evaluate the effects of the number of segments $S$ and the initialization strategy for the $\varepsilon$ levels $\varepsilon_0, \varepsilon_1, \ldots, \varepsilon_S$. All experiments are performed on the most challenging benchmark instance, \texttt{a280\_n279\_bcs}. We vary the number of segments with $S \in \{5, 10, 15\}$. In addition, we compare different initialization strategies for the $\varepsilon$-levels. Specifically, in contrast to the equal initialization strategy, where $\varepsilon_s = g_{\min} + \frac{s}{S}(g_{\max} - g_{\min})$, we also consider a random initialization scheme defined as $\varepsilon_s = g_{\min} + t_s (g_{\max} - g_{\min})$, where $t_0 \le t_1 \le \cdots \le t_S$ are randomly sampled from a uniform distribution over $[0,1]$. As illustrated in Fig.~\ref{table: epsilon level}, the solution quality exhibits a slight improvement when $S = 15$ under the equal initialization strategy, while all approaches converge by the second iteration. This behavior stems from the effectiveness of QA in producing high-quality near-optimal solutions with a single execution. Furthermore,
when the number of segments is set to $S=10$, the resulting HV is significantly higher than that obtained with $S=5$. In contrast, increasing the number of segments to $S=15$ yields only marginal additional improvement. Therefore, $S=10$ provides a favorable trade-off between computational efficiency and solution quality. Furthermore, across all values of $S$, the equal initialization strategy consistently outperforms the randomized initialization. This advantage arises because equal initialization enables a more uniform exploration of the feasible space and maintains appropriate spacing along the second objective. In contrast, random initialization often generates solutions that are either overly clustered or widely dispersed, resulting in less effective coverage of the non-dominated region.

% Specifically, this solution is selected from the Pareto front of the benchmark instance a280\_n279\_bcs, corresponding to the point with the lowest value for the second objective. We focus on the solution derived from the final iteration of the subproblem-solving process. As illustrated in Fig.\ref{fig:LEAcompare}(a), the thief’s initial velocity gradually decreases over time. In contrast, when the LEA search is applied, the thief's velocity remains stable after roughly one-third of the tour. This improvement is particularly significant, as it leads to a reduced travel time for the thief during the later stages of the tour, which is further analyzed in Fig.\ref{fig:LEAcompare}(b). 

% \redd{doan nay co the them 1 hinh ve pareto front nua de visualize truoc va sau cai LEA search, nhung nhu vay se khien minh bi lo mat la moi vong lap chi thuc hien 1 lan}

\vspace{-0.3cm}
\section{Conclusion}\label{conclusion}
In this paper, we addressed the Bi-Objective Traveling Thief Problem (BI-TTP), which jointly optimizes two conflicting objectives: minimizing travel time and maximizing knapsack profit. To solve this challenging problem, we proposed a novel hybrid quantum-classical framework that combines the $\varepsilon$-constraint method with quantum annealing (QA). By first determining the feasible range of knapsack profits using QA, we systematically partitioned this range into uniform $\varepsilon$ intervals, each defining a subproblem with a constrained second objective. These subproblems were reformulated into QUBO models through an auxiliary-variable transformation, enabling direct solution via the D-Wave hybrid CQM solver. Furthermore, we introduced a domain-specific heuristic, called LEA search, to refine QA-generated solutions. Extensive numerical experiments on benchmark instances demonstrated that our method efficiently approximates the Pareto front and significantly reduces computation time compared to classical multi-objective evolutionary algorithms. Looking forward, with advances in quantum hardware, particularly in the number of qubits and QPU reliability, the proposed framework holds promise for solving large-scale, bi-objective combinatorial problems more efficiently.

\vspace{-0.25cm}
\section{Appendix}
\subsection{Proof of Lemma~\ref{Lemma:convert}}\label{appendxi: proof convert}
\begin{proof}
To prove the equivalence between problems \eqref{prob:brief_prob} and \eqref{prob:convert}, we show that they share the same optimal solution and optimal value. Let $\tron{\mathbf{X}^*,\mathbf{Z}^*,\mathbf{b}^*}$ be the global optimal solution for the problem in \eqref{prob:convert}. Since $0 < b_j^* \le B_j\tron{\mathbf{X}^*,\mathbf{Z}^*}$, we have 
\begin{equation}\label{lemma1_1}
    \sum\nolimits_{j=1}^N \dfrac{A_j\tron{\mathbf{X}^*,\mathbf{Z}^*}}{b_j^*} \ge \sum\nolimits_{j=1}^N \dfrac{A_j\tron{\mathbf{X}^*,\mathbf{Z}^*}}{B_j\tron{\mathbf{X}^*,\mathbf{Z}^*}}.
\end{equation}
On the other hand, let $\widetilde{b}_j = B_j\tron{\mathbf{X}^*,\mathbf{Z}^*},\forall \, j = \overline{1,N}$. Then $(\mathbf{X}^*,\mathbf{Z}^*,\widetilde{\mathbf{b}})$ is a feasible solution of problem \eqref{prob:convert}. Since $\tron{\mathbf{X}^*,\mathbf{Z}^*,\mathbf{b}^*}$ is the global optimal solution of problem \eqref{prob:convert}, it follows that 
\begin{equation}\label{lemma1_2}
    \sum\nolimits_{j=1}^N \dfrac{A_j\tron{\mathbf{X}^*,\mathbf{Z}^*}}{b^*_j} \le  \sum\nolimits_{j=1}^N \dfrac{A_j\tron{\mathbf{X}^*,\mathbf{Z}^*}}{\widetilde{b}_j} = \sum\nolimits_{j=1}^N \dfrac{A_j\tron{\mathbf{X}^*,\mathbf{Z}^*}}{B_j\tron{\mathbf{X}^*,\mathbf{Z}^*}}
\end{equation}
From \eqref{lemma1_1} and \eqref{lemma1_2}, we deduce that 
\begin{equation}
    \sum\nolimits_{j=1}^N \dfrac{A_j(\mathbf{X}^*,\mathbf{Z}^*)}{b_j^*} = \sum\nolimits_{j=1}^N \dfrac{A_j(\mathbf{X}^*,\mathbf{Z}^*)}{B_j(\mathbf{X}^*,\mathbf{Z}^*)}
\end{equation}
This result implies that $b^*_j = B_j\tron{\mathbf{X}^*,\mathbf{Z}^*},\forall \, j = \overline{1,N}$. For any $(\mathbf{X}^0,\mathbf{Z}^0) \in \mathcal{D}$, let $b^0_j = B_j(\mathbf{X}^0,\mathbf{Z}^0),\forall \, j = \overline{1,N}$. Then $(\mathbf{X}^0,\mathbf{Z}^0,\mathbf{b}^0)$ is a feasible solution of problem \eqref{prob:convert}, leading to 
\begin{align}
    &\sum_{j=1}^N \dfrac{A_j\tron{\mathbf{X}^*,\mathbf{Z}^*}}{B_j\tron{\mathbf{X}^*,\mathbf{Z}^*}} \notag\\
    =\, & \sum_{j=1}^N \dfrac{A_j(\mathbf{X}^*,\mathbf{Z}^*)}{b_j^*} \le \sum_{j=1}^N \dfrac{A_j(\mathbf{X}^0,\mathbf{Z}^0)}{b_j^0} =  \sum_{j=1}^N \dfrac{A_j(\mathbf{X}^0,\mathbf{Z}^0)}{B_j(\mathbf{X}^0,\mathbf{Z}^0)}.
\end{align}
Therefore, $(\mathbf{X}^*, \mathbf{Z}^*)$ is also the global optimal solution of problem~\eqref{prob:brief_prob}, and both problems~\eqref{prob:brief_prob} and~\eqref{prob:convert} attain the same minimum objective value. The converse follows by similar reasoning, thereby completing the proof.
\end{proof}
% Conversely, let $\tron{\mathbf{X}^*,\mathbf{Z}^*}$ be the optimal solution of problem \eqref{prob:brief_prob} and $\hat{b}_j = B_j\tron{\mathbf{X}^*,\mathbf{Z}^*}$ for all $j = 1,2,\ldots,N$. Since $B_j(\mathbf{X}, \mathbf{Z}) \ge b_j$ for all $j = 1,2,\ldots,N$, we have 
% \begin{align}
%     &\sum_{j=1}^N \dfrac{A_j\tron{\mathbf{X}^*,\mathbf{Z}^*}}{\hat{b}_j} \\
%     =\, & \sum_{j=1}^N \dfrac{A_j\tron{\mathbf{X}^*,\mathbf{Z}^*}}{B_j\tron{\mathbf{X}^*,\mathbf{Z}^*}} \le \sum_{j=1}^N \dfrac{A_j\tron{\mathbf{X},\mathbf{Z}}}{B_j\tron{\mathbf{X},\mathbf{Z}}} \le \sum_{j=1}^N \dfrac{A_j\tron{\mathbf{X},\mathbf{Z}}}{b_j},  
% \end{align}
% for all $\tron{\mathbf{X},\mathbf{Z}} \in \mathcal{D}$. Moreover, since $(\mathbf{X}^*,\mathbf{Z}^*,\hat{\mathbf{b}})$ is a feasible solution of \eqref{prob:convert}, it also becomes the global optimal solution of problem \eqref{prob:convert}. As a result, the minimum value of problem \eqref{prob:convert} is $\displaystyle\sum_{j=1}^N \dfrac{A_j\tron{\mathbf{X}^*,\mathbf{Z}^*}}{B_j\tron{\mathbf{X}^*,\mathbf{Z}^*}}$, which is the same as problem \eqref{prob:brief_prob}. 
\subsection{Proof of Lemma~\ref{coroll1}}\label{appendix: proof coroll1}
\begin{proof}
Suppose $(\mathbf{X}^{(t)}, \mathbf{Z}^{(t)}, \mathbf{b}^{(t)})$ is the solution identified at the $t$-th iteration. At the $(t+1)$-th iteration, each element of the vector $\mathbf{b}^{(t+1)}$ is updated as
    \begin{equation}\label{coroll1.1}
        b_i^{(t+1)} = Wv_{\max} - W^{(t)}_i(v_{\max} - v_{\min}), \forall \, i=\overline{1,N}, 
    \end{equation}
    The constraint at the $t$-th iteration indicates that $b^{(t+1)}_i  \ge b^{(t)}_i$, while at the $t+1$-th iteration, it becomes 
    \begin{equation}\label{coroll1.2}
        Wv_{\max} - W^{(t+1)}_i (v_{\max} - v_{\min}) \ge b^{(t+1)}_i, \forall \, i = \overline{1,N}. 
    \end{equation}
    From equations~\eqref{coroll1.1} and~\eqref{coroll1.2}, it follows that $W^{(t)}_i \ge W^{(t+1)}_i$ for all $t$. Consequently, we have
    \begin{equation}
        \di\mathbf{\Theta}_{t+1} = \mathbf{\Theta}_t \di\bigcap_{i=1}^N \nhon{W^{(t+1)}_i(\mathbf{X}^{(t+1)}, \mathbf{Z}^{(t+1)}) \le W^{(t)}_i} \subseteq \mathbf{\Theta}_t, 
    \end{equation}
    where $\mathbf{\Theta}_t$ denotes the feasible region at iteration $t$. Therefore, the update guaranties a monotonic reduction in the search space. Since the overall feasible space is finite, the iterative procedure always converges.
    \end{proof}
\subsection{Detail description of our proposed \textit{Later and Enough Accumulation} heuristic search} \label{appendix: LEA description}

The proposed LEA method operates in two phases: the first phase shifts the item selection toward the later cities in the tour, and the second phase iteratively reduces the total weight of the selected items while ensuring compliance with the knapsack capacity constraint. To facilitate this process, the picking plan $\mathbf{Z}^{\text{QA},(t)}$ is flattened into a single vector as follows:
\begin{equation}
    \mathbf{Z}^{\mathrm{LEA}} = \vuong{\mathbf{z}_{\alpha_1}^{\text{QA},(t)}; \mathbf{z}_{\alpha_2}^{\text{QA},(t)};\ldots;\mathbf{z}_{\alpha_N}^{\text{QA},(t)}}.
\end{equation}
\begin{algorithm}[t]
    \caption{Later and Enough Accumulation (LEA)}
    \label{alg:LEA}
\hspace*{\algorithmicindent}\textbf{Input:}
the traveling strategy $\mathbf{X} = \mathbf{{X}}^{\text{QA},(t)}$, the corresponding picking plan $\mathbf{Z} = \mathbf{{Z}}^{\text{QA},(t)}$ \\
\hspace*{\algorithmicindent}\textbf{Output:}  New picking plan $\mathbf{Z^{\text{LEA}}}$
\begin{algorithmic}[1] 
    \State $\mathbf{Z}^{\text{LEA}} \gets \mathbf{Z}$
    \For {$1 \le p < M$} \Comment{Phase 1}
        \For {$p+1 \le q \le M$}
            \If {$\tron{\vuong{\mathbf{Z}^{\text{LEA}}}_p, \vuong{\mathbf{Z}^{\text{LEA}}}_q} = (1,0)$}
                \State $\mathbf{Z}^{\text{swap}} \gets \mathbf{Z}^{\text{LEA}}, \textbf{Swap}\tron{\vuong{\mathbf{Z}^{\text{swap}}}_p,\vuong{\mathbf{Z}^{\text{swap}}}_q}$
                \If {$f\tron{\mathbf{X},\mathbf{Z}^{\text{swap}}} < f(\mathbf{X}, \mathbf{Z}^{\text{LEA}})\text{ \& } (\mathbf{X},\mathbf{Z}^{\text{swap}}) \in \mathcal{D} $}
                    \State $\mathbf{Z}^{\text{LEA}} \gets \mathbf{Z}^{\text{swap}} $
                \EndIf
            \EndIf
        \EndFor
    \EndFor
    \For {$1 \le p \le M$}  \Comment{Phase 2}
        \If {$\vuong{\mathbf{Z}^{\text{LEA}}}_p =1$}
            \State $\mathbf{Z}^{\text{temp}} \gets \mathbf{Z}^{LEA}, \vuong{\mathbf{Z}^{\text{temp}}}_p \gets 0$
            \If {$\tron{\mathbf{X}, \mathbf{Z}^{\text{temp}}} \in \mathcal{D}$}
                \State $\mathbf{Z}^{\text{LEA}} \gets \mathbf{Z}^{\text{temp}}$
            \EndIf
        \EndIf
    \EndFor
    \State \Return $\mathbf{{Z}}^{(t)} = \mathbf{Z}^{\text{LEA}}$ and $\mathbf{{X}}^{(t)} = \mathbf{X}^{\text{QA},(t)}$.
\end{algorithmic}
\end{algorithm}
Here, we recall that $c_{\alpha_i}$ denotes the $i$-th city visited in the tour. In the first phase, the proposed LEA algorithm iteratively loops through all elements of the picking vector $\mathbf{Z}^{\mathrm{LEA}}$. At the $p$-th iteration, if $[\mathbf{Z}^{\mathrm{LEA}}]_p = 1$, the algorithm searches for an index $q > p$ where $[\mathbf{Z}^{\mathrm{LEA}}]_q = 0$. If replacing $[\mathbf{Z}^{\mathrm{LEA}}]_p$ with $0$ and $[\mathbf{Z}^{\mathrm{LEA}}]_q$ with $1$ maintains the feasibility and leads to a reduction in total travel time, the change is accepted. This selection strategy encourages the thief to prioritize item collection near the end of the tour, while ensuring that the total number of selected items remains unchanged. 
In the second phase, the heuristic iterates through all elements of $\mathbf{Z}^{\text{LEA}}$ obtained from the first phase. For each index $p$ where $\vuong{\mathbf{Z}^{\text{LEA}}}_p = 1$, the item is deselected (i.e., $\vuong{\mathbf{Z}^{\text{LEA}}}_p = 0$) if the  solution still satisfies the knapsack constraint and maintains Pareto optimality. The detailed process is  in Algorithm \ref{alg:LEA}.
\bibliographystyle{IEEEtran}
\bibliography{IEEEabrv,references}

%\bibliography{refs}

\end{document}